\documentclass[12pt, draftclsnofoot, onecolumn]{IEEEtran}
\usepackage{fullpage,graphicx,psfrag,verbatim,color}
\DeclareGraphicsExtensions{.jpg,.pdf,.jpeg,.png}
\usepackage[small,bf]{caption}
\usepackage{amsmath}
\usepackage{amsfonts,bm,amsthm} 
\usepackage{bbm}
\usepackage{mathtools}
\usepackage{amssymb}
\usepackage{algorithm}  
\usepackage{algorithmic} 
\usepackage{subfig}
\usepackage{geometry}
\newcommand{\field}[1]{\mathbb{#1}}
\newcommand{\C}{\field{C}}
\newcommand{\R}{\field{R}}
\newcommand{\T}{\text{T}}
\newcommand{\CT}{\text{H}}
\newcommand{\abs}[1]{\left\lvert#1\right\rvert}
\newcommand{\norm}[1]{\left\lVert#1\right\rVert}
\newcommand{\vect}[1]{\mathbf{#1}}
\newcommand{\mat}[1]{\mathbf{#1}}

\newcommand{\Expect}{\mathop{\mathbb E{}}}

\DeclareMathOperator{\Tr}{Tr}

\DeclareMathOperator{\Diag}{Diag}
\DeclareMathOperator{\dl}{dl}
\DeclareMathOperator{\ul}{ul}
\DeclareMathOperator{\mmse}{mmse}

\DeclareMathOperator{\RZF}{RZF}
\DeclareMathOperator{\pilot}{pilot}
\DeclareMathOperator{\Vectorize}{Vec}
\DeclareMathOperator{\wpone}{w.p.1}
\DeclareMathOperator{\Blk}{Blk}
\DeclareMathOperator{\BlkTran}{BlkTran}

\newcommand{\roundBrack}[1] { \left(#1\right)   }   
\newcommand{\squareBrack}[1]{ \left[#1\right]   }   
\newcommand{\curlyBrack}[1] { \left\{#1\right\} }   
\newcommand{\ROMAN}[1]{\uppercase\expandafter{\romannumeral#1}}

\newcommand{\CN}[1]{\mathcal{CN}\roundBrack{#1}}

\theoremstyle{definition}
\newtheorem{defn}{Definition}
\newtheorem{assumption}{Assumption}
\newtheorem{theorem}{Theorem}
\newtheorem{lemma}{Lemma}

\begin{document}
	
	\title{Randomized Channel Sparsifying Hybrid Precoding for FDD Massive MIMO Systems}
	

	\author{\IEEEauthorblockN{ Chang Tian\IEEEauthorrefmark{1},
		An Liu\IEEEauthorrefmark{2},
		Mahdi Barzegar Khalilsarai\IEEEauthorrefmark{3},
		Giuseppe Caire\IEEEauthorrefmark{3},
		Wu Luo\IEEEauthorrefmark{1}, and
		Minjian Zhao\IEEEauthorrefmark{2} }
	
	\thanks{ \IEEEauthorrefmark{1} State Key Laboratory of Advanced Optical Communication Systems and Networks, Department of Electronics, Peking University ( \{tianch, luow\}@pku.edu.cn ). }
	\thanks{ \IEEEauthorrefmark{2}College of Information Science and Electronic Engineering, Zhejiang University ( \{anliu, mjzhao\}@zju.edu.cn )  \textit{Corresponding author: An Liu}. }
	\thanks{ \IEEEauthorrefmark{3}Department of Telecommunication Systems, Technical University of Berlin ( \{m.barzegarkhalilsarai, caire\}@tu-berlin.de ).  }
}
	
	
	\maketitle

	\begin{abstract}
    We propose a novel randomized channel sparsifying hybrid precoding (RCSHP) design to reduce the signaling overhead of channel estimation and the hardware cost and power consumption at the base station (BS), in order to fully harvest benefits of frequency division duplex (FDD) massive multiple-input multiple-output (MIMO) systems. RCSHP allows time-sharing among multiple analog precoders, each serving a compatible user group.  The analog precoder is adapted to the channel statistics to properly sparsify the channel for the associated user group, such that the resulting effective channel (product of channel and analog precoder) not only has enough spatial degrees of freedom (DoF) to serve this group of users, but also can be accurately estimated under the limited pilot budget. The digital precoder is adapted to the effective channel based on the duality theory to facilitate the power allocation and exploit the spatial multiplexing gain.  We formulate the joint optimization of the time-sharing factors and the associated sets of analog precoders and power allocations as a general utility optimization problem, which considers the impact of effective channel estimation error on the system performance.  Then we propose an efficient stochastic successive convex approximation algorithm to provably obtain Karush-Kuhn-Tucker (KKT) points of this problem.
\end{abstract}
	

	\section{Introduction}
	
	\IEEEPARstart{T}{he} fifth generation (5G) wireless network has been envisioned to be ultra reliable, resource-efficient, low-latency and secure, thanks to the developments of 5G key techniques, such as massive multiple-input multiple-output (MIMO), millimeter wave (mmWave), ultra-dense heterogeneous networks and mobile edge computing \cite{2018LowLatency5G,2019Multi-Functional5G}. Among these advanced techniques, massive MIMO is considered as one of the most promising ways to improve the spectral efficiency. Meanwhile, a well-known fact is that the frequency division duplex (FDD) protocol dominates current wireless cellular systems \cite{2014OpenLoopClosedLoop,2015SptllyCommonSparsity}. Also motivated by spectrum regulation issues, there is a huge commercial interest in enabling the FDD massive MIMO to be compatible with current wireless networks \cite{2015SptllyCommonSparsity}. Therefore, designing practical and efficient precoding schemes for FDD massive MIMO systems is necessary.
	
	The traditional pure digital precoding schemes require one radio frequency (RF) chain for each antenna, which leads to the huge hardware cost and power consumption of the massive MIMO base station (BS). As a result, hybrid precoding, where a high-dimensional analog precoder is connected to a reduced-dimensional digital precoder with a limited number of RF chains, has been proposed to address this issue. Early works mainly focus on studying fast-timescale hybrid precoding (FHP) schemes, in which both analog and digital precoders are adapted to the instantaneous channel state information (CSI). For example, in \cite{2014Phased-ZF}, a low-complexity phased-zero-forcing hybrid precoding scheme is proposed for massive MIMO systems. A sparse precoding and combining scheme based on the concept of orthogonal matching pursuit is proposed for single-user mmWave MIMO systems in \cite{2014_SparsePrecoding_Heath}. Then in \cite{2015LimitedFeedbackmmwave}, a limited feedback hybrid precoding scheme is proposed for multi-user mmWave systems. However, FHP schemes usually induce a large CSI signaling overhead due to the acquisition of real-time CSI. Moreover, different analog precoders need to be implemented for different subcarriers since different subcarriers may have different real-time CSI \cite{2016ImpactOfCSI}. 
		
	To overcome the above disadvantages of FHP, \cite{2013JSDM} and \cite{2014ProposeTHP} propose the two-timescale hybrid precoding (THP) scheme, where the analog precoder is adaptive to the channel statistics to achieve the array gain and the digital precoder is adaptive to the reduced-dimensional effective CSI (product of channel and analog precoder) to achieve the spatial multiplexing gain. Based on the above insights, \cite{2018SSCA-THP} proposes an online algorithmic framework for general THP optimization problems. A THP scheme for downlink multi-cell massive MIMO systems is proposed in \cite{2019RTHP} based on the two-stage stochastic optimization. THP can significantly reduce the CSI signaling overhead because it does not require the knowledge of real-time high-dimensional CSI. Furthermore, since the channel statistics are approximately same on different subcarriers \cite{2016ImpactOfCSI,2017ExploitSpatialCovariance}, THP only needs one analog precoder to cover all subbands, which makes it more attractive in practice due to the low implementation cost \cite{2016ImpactOfCSI}.
		
	In existing works on (two-timescale) hybrid precoding, the analog precoder is designed to optimize the downlink transmission performance by assuming the perfect knowledge of effective CSI. However, in practice, the analog precoder can influence the downlink transmission performance not only in directly affecting the spatial degrees of freedom (DoF) of effective CSI, but also in directly affecting the quality of effective CSI estimation. The effect of analog precoding on the quality of effective CSI estimation is usually ignored in existing literatures. To achieve a better downlink transmission performance in practice, the optimization of analog precoding should also take into account the impact of analog precoding on the effective CSI estimation error. This issue is more challenging in a FDD massive MIMO downlink transmission scenario, in which the channel reciprocity can not be exploited and the number of assigned pilot symbols is limited for the consideration of reducing the amount of radio resources consumed by the CSI signaling overhead. Some compressive sensing (CS)-based methods have been proposed \cite{2014J-OMP,2017AdaptivePilot-MU,2018Yuen,2019Yuen} to guarantee the high-quality CSI at the BS by only utilizing a small number of pilot symbols.  However, all these works heavily rely on the assumption that  propagation channels have intrinsic sparse properties, which may not be satisfied in practice, especially for systems operating at the sub-6GHz frequency \cite{2015sub_6G}. 
	
     \cite{2018ActiveChannelSpar} has partially addressed the above issue by proposing a two-stage digital precoding scheme, in which a zero-forcing (ZF) precoder is connected with a sparsifying precoder and both are implemented in the digital domain. The sparsifying precoder is designed to select beams\footnote{ The selected beams refer to the selected angular directions used to transmit information to users. The number of selected beams for each user can be reflected by the rank of each user's effective channel after channel sparsifying.  } of each selected user, such that the effective channel (product of channel and sparsifying precoder in \cite{2018ActiveChannelSpar}) of each selected user is sparse enough and thus enables to achieve a good effective channel estimation quality, by utilizing a limited number of assigned pilot symbols. This beam-user selection procedure is referred to as \textit{active channel sparsification}. By applying the \textit{active channel sparsification}, the spatial DoF of effective CSI is maximized to achieve a good spectral efficiency. However, there are some drawbacks in \cite{2018ActiveChannelSpar}: 1) The proposed sparsifying precoding scheme implements both precoders in the digital domain. Reducing the number of RF chains is not taken into account, and thus leads to a relatively large hardware cost and power consumption. 2) It only considers the sum rate maximization while neglecting the fairness among users. When the number of users is larger than the available spatial DoF, only a subset of users will be scheduled for transmission over a large number of channel coherence intervals. There still lacks an efficient user grouping/selection method to achieve a better tradeoff between the sum throughput and fairness under \textit{active channel sparsification}. 3) The channel sparsifying precoder is not designed to directly optimize the throughput performance, but is designed based on a heuristic criteria. 
	
	In this paper, we consider a practical FDD massive MIMO downlink transmission scenario, in which the channel environment may not be sparse and only a limited number of pilot symbols is available. A randomized channel sparsifying hybrid precoding (RCSHP) design is proposed to strike a balance between the spatial DoF and the estimation error of effective CSI, such that the overall downlink transmission performance can be improved with the reduced hardware cost and power consumption. Specifically, RCSHP allows time-sharing among multiple analog precoders, each serving a compatible user group. The time-sharing factors and the associated analog precoders and power allocations are adapted to the channel statistics to properly sparsify the channel for each associated user group, such that the effective channel not only has enough spatial DoF to serve this group of users, but also can be sparse enough and accurately estimated under the limited pilot budget. The digital precoder is designed based on the uplink-downlink duality and is adapted to the effective channel to achieve the spatial multiplexing gain. The main contributions are summarized below.
	\begin{itemize}
		\item \textbf{Randomized Analog Precoding and Power Allocation Scheme:} This scheme allows a more refined control on the analog precoders and power allocations, such that a specific analog precoder can be used to cover a user group. By time-sharing among the multiple analog precoders and power allocations, all users can enjoy a non-zero average data rate, achieving a better tradeoff between the sum throughput and fairness. The user selection and the beam (angular direction) selection for each selected user's effective channel (\textit{active channel sparsification}) can be automatically achieved by jointly optimizing the time-sharing factors and the associated analog precoders and power allocations, which is more robust with respect to various types of propagation environments \footnote{  The proposed RCSHP scheme can also take advantage of the channel sparsity to achieve a better performance. When the channel is naturally sparse, the optimized sets of analog precoder will tend to concentrate on the few active channel paths to fully harvest the spatial multiplexing and array gain. Thus the proposed design is also suitable for the mmWave massive MIMO system with more sparse channels, which is a promising application scenario for hybrid precoding.  }. 
		
		\item \textbf{Duality-based Digital Precoder:} We obtain the duality-based digital precoder from a virtual uplink reception problem based on the minimum mean square error (MMSE) rule, by exploiting the duality that the precoding concepts designed for the downlink transmission can carry over to the corresponding virtual uplink reception. The proposed duality-based digital precoder has a similar complexity as that of the regularized zero forcing (RZF) precoder, but is a smooth function of the power allocation (as will be explained in Section \ROMAN{2}-D), leading to a tractable power allocation optimization formulation.
		
		\item \textbf{General Utility Optimization:} The proposed RCSHP is formulated as a general utility optimization problem, including sum rate maximization and proportional fairness (PFS) utility maximization as special cases, such that the proposed RCSHP can cover more application scenarios. However, this incurs a challenging non-convex stochastic optimization problem. To address this problem, we propose an efficient stochastic successive convex approximation (SSCA) algorithm called SSCA-RCSHP, and also establish the convergence of SSCA-RCSHP to KKT points of the general utility optimization problem. 
	\end{itemize}
	The rest of the paper is organized as follows. System model is presented in Section \ROMAN{2}. We formulate the RCSHP design as a general utility optimization problem in Section \ROMAN{3}. The proposed SSCA-RCSHP algorithm, and the associated convergence and complexity analysis are presented in Section \ROMAN{4} and \ROMAN{5}, respectively. Further, simulation results are given in Section \ROMAN{6}, and we conclude the paper in Section \ROMAN{7}.
	
	\textit{Notations}: $ \Diag\roundBrack{\vect{a}} $ represents a diagonal matrix whose diagonal elements form the vector $ \vect{a} $. $ \squareBrack{\mat{M}}_{i.} $, $ \squareBrack{\mat{M}}_{.i} $ and $ \squareBrack{\mat{M}}_{ij} $ denote the $ i $-th row, $ i $-th column and $ \roundBrack{i,j} $-th element of matrix $ \mat{M} $, respectively. $ \otimes $ denotes Kronecker product and $ \circ $ denotes Hadamard product. $ \Tr(\cdot) $, $(\cdot)^{*}$ $(\cdot)^{\T}$, $(\cdot)^{\CT}$, $ \norm{\cdot}_1 $, $ \norm{\cdot}_2 $, $ \norm{\cdot}_F $, $ \vect{1} $ and $ \mat{I} $ denote trace, conjugate, transpose, conjugate transpose, $ l_1 $ norm, $ l_2 $ norm, Frobenius norm, all-one vector and identity matrix, respectively. Let $ \Vectorize\curlyBrack{\mat{M}} $ denote the vectorization of matrix $ \mat{M} $ and $ \Re\squareBrack{\mat{M}} $ denote the real part of a complex matrix $ \mat{M} $. 

	\section{System Model}
	\subsection{FDD Massive MIMO Downlink with Hybrid Precoding}
	Consider a multi-user FDD massive MIMO downlink system with one BS serving $ K $ single-antenna users. For clarity, we focus on a narrowband system with a flat block fading channel, where the channel coefficients are assumed to be constant over a block containing $ T $ symbols, but the proposed design can be easily extended to the wideband system \footnote{ In this paper, the analog precoder is adapted to the channel statistics only, and thus the same analog precoder will be used on different subcarriers in a wideband system. However, the digital precoders on different subcarriers can be different.  }. Assume the channel changes over blocks according to certain distribution, e.g., $ \vect{h}_k\sim\CN{\vect{0},\mat{C}_k}, \forall k $, where $ \vect{h}_k\in\C^M $ is the channel for user $ k $ and $ \mat{C}_k = \Expect\curlyBrack{\vect{h}_k \vect{h}_k^{\CT}}\in\C^{M\times M} $ is the channel covariance of user $ k $, for the convenient design consideration as similar to \cite{2018ActiveChannelSpar}. The BS employs the hybrid precoding architecture, as illustrated in Fig. \ref{HybridPrecoding}, which is equipped with $ M $ antennas and $ S $ transmit RF chains, where $ S\ll M  $. It is necessary to clarify that the proposed scheme considers fairness in terms of the long-term average throughput. Thus the number of served users $ K $ is allowed to be more than the number of RF chains $ S $. However, the number of active users at each time slot is usually less than $ S $ \footnote{ Actually, we do not add any explicit constraint to restrict that the number of active users must be less than $ S $. However, supporting more than $ S $ active users with $ S $ RF chains will cause large multi-user interference, which is usually not optimal.  }. Moreover, a two-timescale hybrid precoder is employed to transmit data streams with limited RF chains. The transmit vector for user $ k $ is given by $ \sqrt{p}_k\mat{F}\vect{g}_k s_k $, where $ \mat{F}\in\C^{M\times S} $ is the  analog precoder, $ \vect{g}_k\in\C^{S\times 1} $ with $ \norm{\mat{F}\vect{g}_k}_2=1 $ is the $ k $-th column vector of the normalized digital precoder $ \mat{G}\in\C^{S\times K} $, $ p_k $ is the transmit power allocated to user $ k $ and $ s_k $ is the unity-power data symbol for user $ k $. The analog precoder $ \mat{F} $ is adaptive to the channel statistics to exploit the array gain and sparsify the channel, and it is usually implemented using an RF phase shifting network\footnote{ In practice, modern implementations are possible to allow for the full analog vector modulation, which means that the analog precoder can be adjusted on both amplitude and phase. The proposed scheme can be easily modified to cover the analog precoder design using full vector modulators, actually just need to add the amplitude as an extra optimization variable.  } \cite{2005RFshiftingNetwork}. Hence, all elements of $ \mat{F} $ have an equal magnitude and can be represented by a phase vector $ \pmb{\theta}\in\squareBrack{0,2\pi}^{MS} $, i.e., $ \squareBrack{\mat{F}}_{ij}=\frac{1}{\sqrt{M}}e^{\sqrt{-1}\theta_{ij}} $, where $ \theta_{ij} $ is the phase of the $ \roundBrack{i,j} $-th element of $ \mat{F} $ and corresponds to the $ \roundBrack{\roundBrack{j-1}M + i }$-th element of $ \pmb{\theta} $. The digital precoder $ \mat{G} $ is adaptive to the estimated effective CSI to achieve the spatial multiplexing gain and mitigate the inter-user interference. Under this setting, the received signal for user $ k $ is
	\begin{equation}\label{ReceivedSignal}
	y_k = \sqrt{p_k}\vect{h}_k^\CT\mat{F}\vect{g}_k s_k + \vect{h}_k^\CT\sum_{i\neq k}\sqrt{p_i}\mat{F}\vect{g}_i s_i + z_k,
	\end{equation}
	where $ z_k\sim\CN{0,1} $  is the normalized additive white Gaussian noise (AWGN).
	\begin{figure}[!t]
		\centering
		\includegraphics[width=4in]{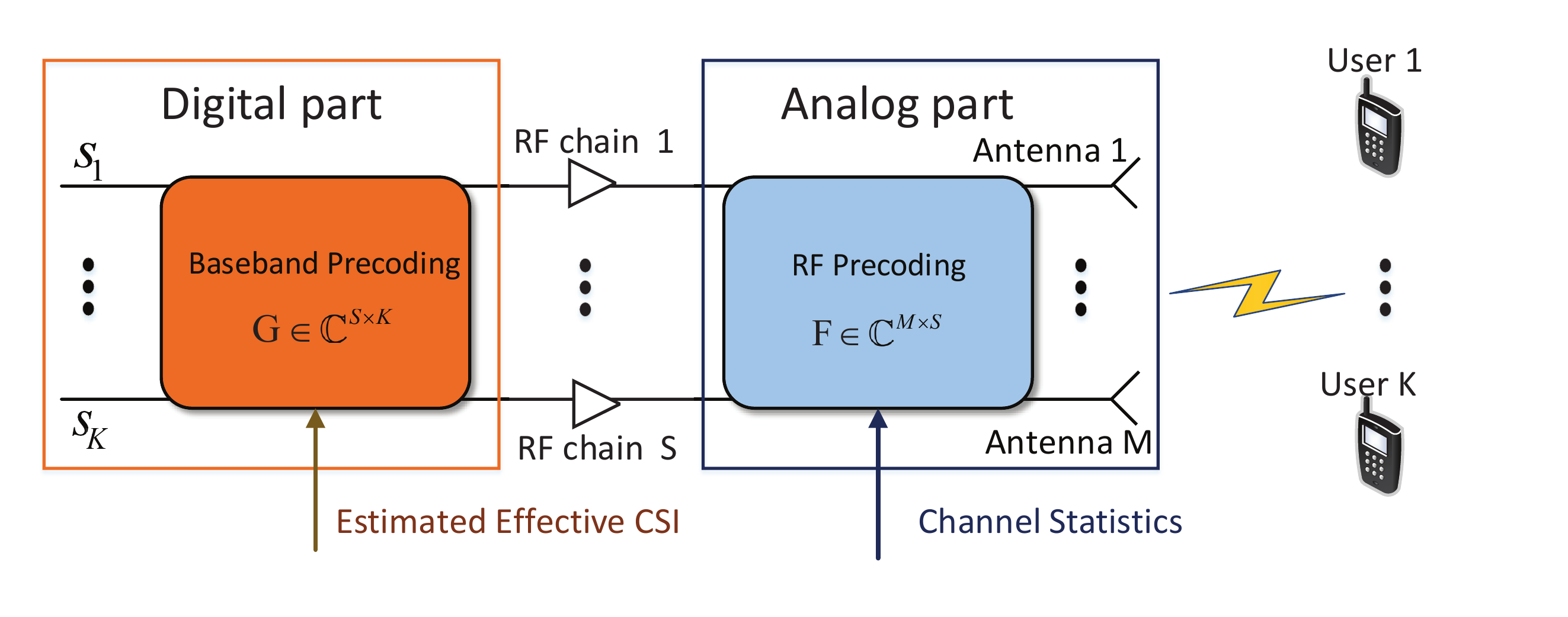}
		\caption{Hybrid precoding architecture.}
		\label{HybridPrecoding}
	\end{figure}
	
	\subsection{Randomized Analog Precoding and Power Allocation Policy}
	There may not always be enough spatial DoF to support the simultaneous transmission to all users. For a fixed analog precoder, it is possible that only a subset of users can be scheduled for transmission over a large number of time slots, when the number of users is larger than the available spatial DoF. Hence, for the fairness consideration, we consider a randomized analog precoding and power allocation policy as defined below, which realizes time-sharing among several analog precoders and power allocations. The analog precoder and power allocation are together referred to as the control variable for conciseness.
	\begin{defn}
		\textit{(Randomized Control Policy): A randomized control policy $ \vect{\Omega}=\curlyBrack{\vect{\Gamma},\vect{q}} $ consists of an aggregated vector of $ L $  control variables $ \vect{\Gamma}\triangleq\squareBrack{\vect{\Gamma}\roundBrack{1}^\T,\dots,\vect{\Gamma}\roundBrack{L}^\T}^\T $ and a probability vector (time-sharing factors) $ \vect{q}\triangleq\squareBrack{q_1,\dots,q_L}^\T $, where the $ l $-th control variable in $ \vect{\Gamma} $ is $ \vect{\Gamma}\roundBrack{l}=\squareBrack{\pmb{\theta}\roundBrack{l}^\T, \vect{p}\roundBrack{l}^\T}^\T $ and $ \vect{\Gamma}$ satisfies $\mathcal{G}=\curlyBrack{\pmb{\theta}\roundBrack{l}\in\squareBrack{0,2\pi}^{MS}, \vect{p}\roundBrack{l}\in\R^{K}_+,  \vect{1}^\T\vect{p}\roundBrack{l}\leq P_{\max}, \forall l } $, and $ \vect{q} $ satisfies $ \mathcal{Q}=\curlyBrack{\vect{q}\in\squareBrack{0,1}^L, \vect{1}^\T\vect{q}=1} $. $ P_{\max} $ is the power budget at the BS. At any time slot, the control variable $ \vect{\Gamma}\roundBrack{l} $ is applied with probability $ q_l $, i.e., the analog precoder and power allocation are respectively given by $ \pmb{\theta}\roundBrack{l} $ and $ \vect{p}\roundBrack{l} $ with probability $ q_l $.}
	\end{defn}
	In the proposed RCSHP, the time-sharing factors and the associated control variables are first jointly optimized according to the channel statistics information at the beginning of each coherence time of channel statistics. Coherence time of channel statistics refers to the time interval that the channel statistics remains unchanged. Then the optimized control policy is applied to time slots of the current coherence time of channel statistics to realize time-sharing among different control variables.  Clearly, choosing a larger $ L $ can lead to a better performance or at least as good as that of a smaller $ L $, since a control policy with a larger $ L $ includes that with a smaller $ L $ as a special case. However, the complexity of the optimization algorithm will also increase with $ L $. As such, we can use $ L $ to control the tradeoff between the performance and complexity. In simulations, we find that a moderately large $ L $ (4 or 5) can already achieve a good performance.
	
	\begin{figure}[!t]
		\centering
		\subfloat[]{\includegraphics[width=2.5in]{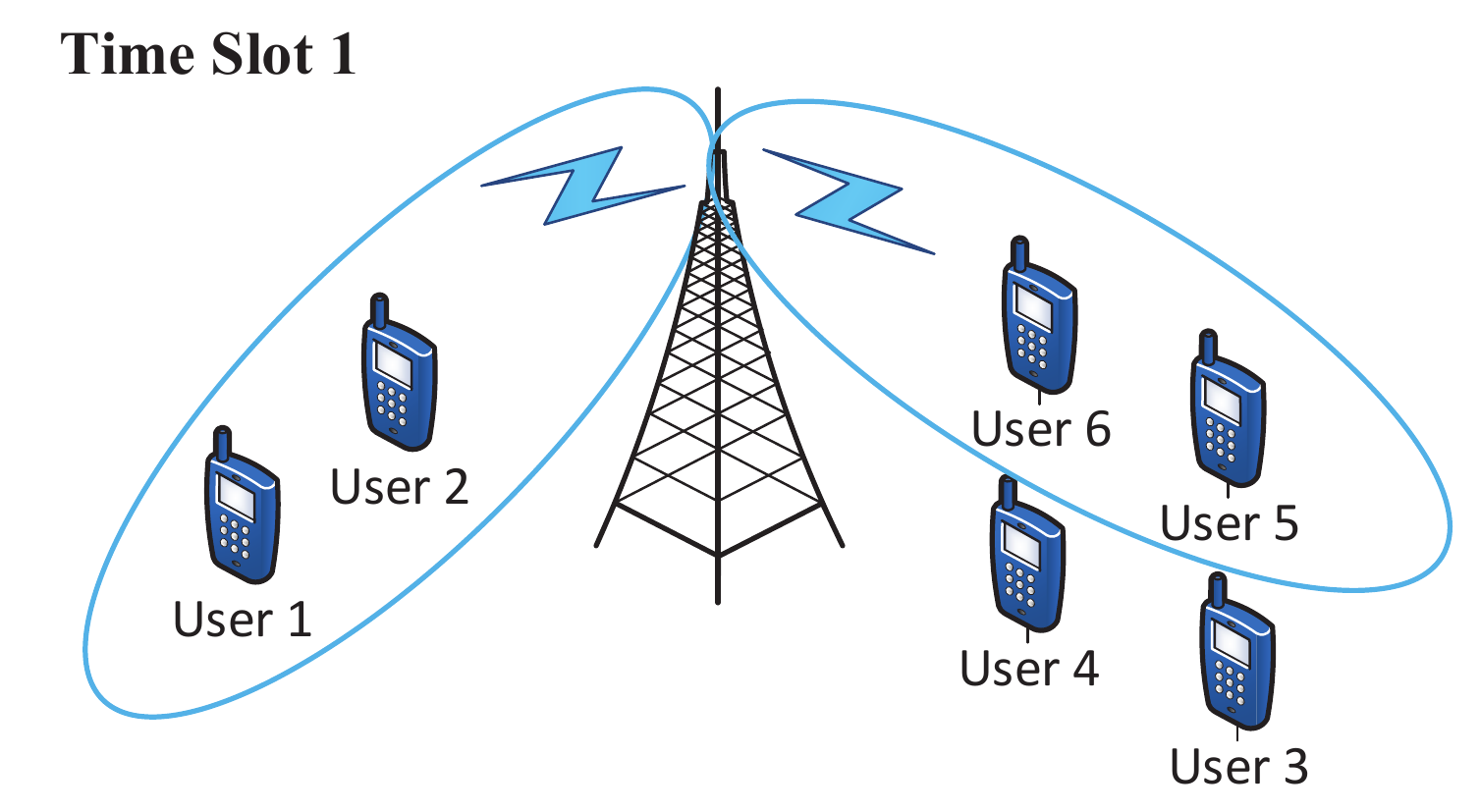}%
			\label{user_group_example_a}}
		\hfil
		\subfloat[]{\includegraphics[width=2.5in]{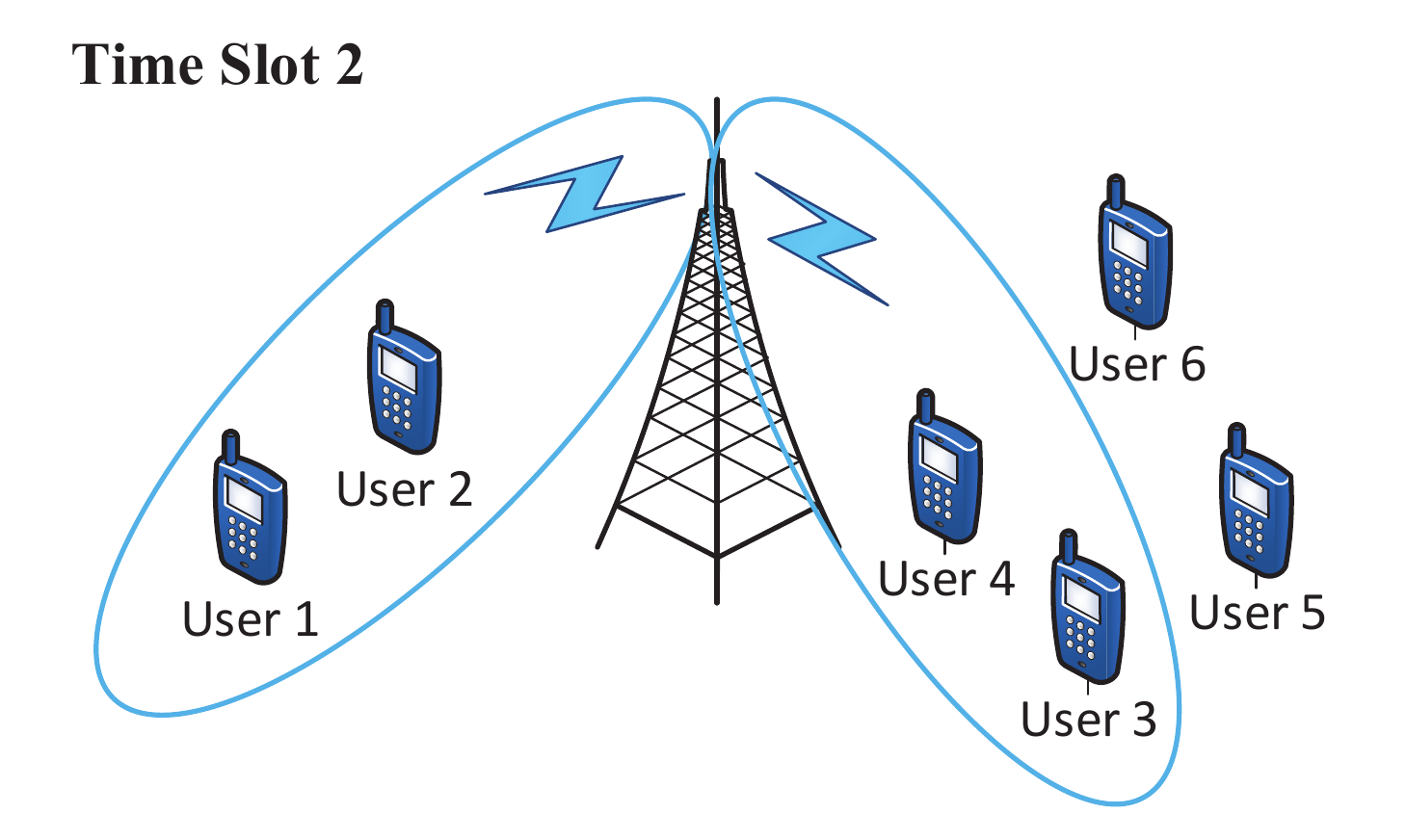}%
			\label{user_group_example_b}}
		\caption{An illustration of the randomized analog precoding and power allocation scheme at two time slots. \textbf{(a)} An example at time slot 1. \textbf{(b)} An example at time slot 2.}
		\label{user_group_example}
	\end{figure}
	
	We use a toy example shown in {\figurename~\ref{user_group_example}} and {\figurename~\ref{OpratingProcess}} to illustrate the proposed RCSHP, which is specified by a set of $ L=2 $ control variables $ \vect{\Gamma}=\squareBrack{\vect{\Gamma}\roundBrack{1}^\T,\vect{\Gamma}\roundBrack{2}^\T}^\T $ and a time-sharing vector $ \vect{q}=\squareBrack{0.4,0.6}^\T $. Specifically, at time slot 1 ({\figurename~\ref{user_group_example_a}}), the associated analog precoder and power allocation at the BS is $ \vect{\Gamma}\roundBrack{1} $ and is compatible with a group of users (user 1,2,5 and user 6), which can be simultaneously scheduled for transmission. The other ``incompatible'' users (user 3 and user 4) can not be scheduled due to the strong inter-user interference. However, at time slot 2 ({\figurename~\ref{user_group_example_b}}), the associated analog precoder and power allocation is $ \vect{\Gamma}\roundBrack{2} $ and is compatible with user 1, 2, 3 and user 4. Therefore, it can be observed from this example that all users can enjoy a non-zero average data rate by time-sharing between these two control variables, achieving a better tradeoff between the sum throughput and fairness. For convenience, we assume that the current coherence time of channel statistics consists of 5 time slots. One possible realization of RCSHP is illustrated in {\figurename~\ref{OpratingProcess}}, where $ \vect{\Gamma}\roundBrack{1}$ and $\vect{\Gamma}\roundBrack{2} $ are used in $ 40\% $ and $ 60\% $ of time slots, respectively.
	\begin{figure}[!t]
		\centering
		\includegraphics[width=4.5in]{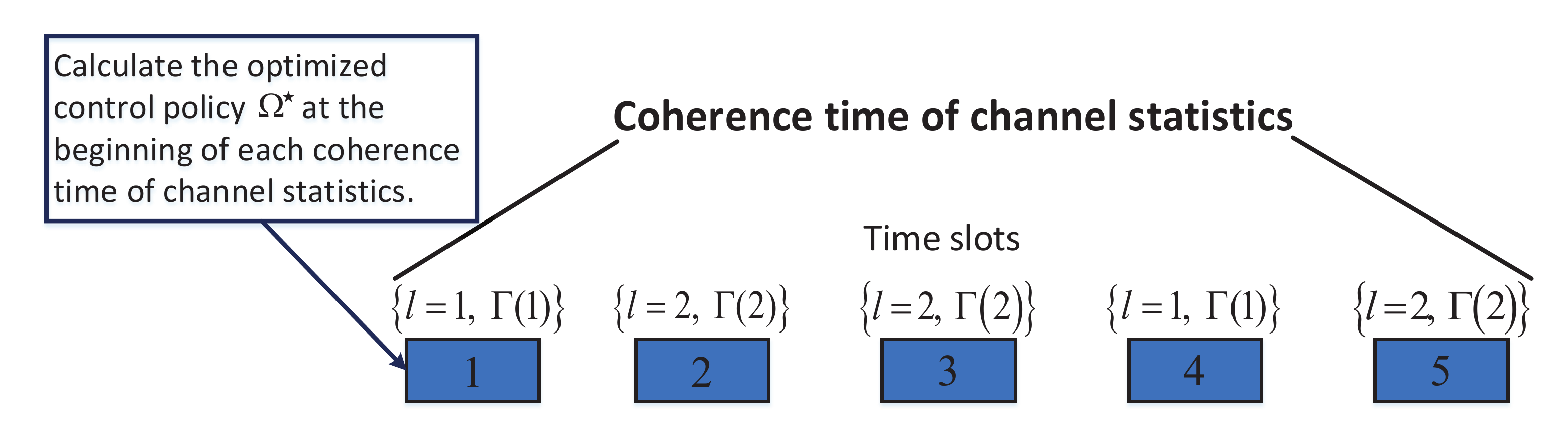}
		\caption{A toy example of RCSHP.}
		\label{OpratingProcess}
	\end{figure}
	
	\subsection{Instantaneous Effective CSI Estimation}
	The effective CSI estimation quality has a significant impact on the downlink transmission performance. We consider a closed-loop scheme to estimate the instantaneous effective CSI, in which the BS sends a sequence of $ T_p $ pilot symbols through the $ S $ inputs of the analog precoder $ \mat{F} $ to all users and each user feeds back its unquantized pilot observation. The estimation of effective CSI for each user is then implemented at the BS. Denote the aggregated common transmit pilot symbols as $ \mat{\Psi}\in\C^{T_p\times S} $. The corresponding pilot observation at the user $ k $ is  
	\begin{equation}
	\vect{y}_k^{\pilot} = \mat{\Psi}\roundBrack{\vect{h}_k^\CT\mat{F}}^\T+\vect{n}_k = \mat{\Psi} \widetilde{\vect{h}}^{*}_k +\vect{n}_k,
	\end{equation}
	where $ \widetilde{\vect{h}}_k=\mat{F}^\CT\vect{h}_k$ is the effective channel of user $ k $ and the aggregated effective channel matrix is given by $ \widetilde{\mat{H}}=\mat{HF}=\squareBrack{\widetilde{\vect{h}}_1,\dots,\widetilde{\vect{h}}_K}^\CT $, and the channel estimation noise is normalized AWGN with distribution $ \vect{n}_k\sim\CN{\vect{0},\mat{I}} $.
	
	Similar to \cite{2018ActiveChannelSpar}, we assume the noiseless analog feedback for clarity consideration, but the proposed scheme can be easily modified to consider the noisy feedback. The BS implements linear minimum mean square error (LMMSE) estimation to estimate the effective CSI, since the estimation quality using LMMSE is already good enough when the number of selected beams of the effective CSI is less than the number of pilot symbols. The LMMSE estimation of the effective channel for user $ k $ is 
	\begin{equation}\label{LMMSE}
	\hat{\widetilde{\vect{h}}}_k = \mat{F}^\CT\mat{C}_k\mat{F} \mat{\Psi }^\T  \roundBrack{ \mat{\Psi }^{*}  \mat{F}^\CT\mat{C}_k\mat{F} \mat{\Psi }^\T  + \mat{I} }^{-1} \roundBrack{ \vect{y}_k^{\pilot} }^{*} .
	\end{equation}

	In our scheme, the channel covariances $ \mat{C}_k, \forall k $ are assumed to be known at the BS. It is reasonable since there are many efficient channel covariance estimation methods in the hybrid precoding architecture. Please refer to \cite{2017SubspaceEsti1,2018SubspaceEsti2} and references therein. Moreover, by exploiting the downlink/uplink angle reciprocity, the downlink channel covariance can be obtained from uplink training pilots even in FDD systems. Please refer to \cite{2018ActiveChannelSpar,2018AngleReciprocity3} and references therein.
	
	\subsection{Duality-based Digital Precoder}
	We propose a duality-based digital precoder by exploiting the duality between the multi-user downlink system and the corresponding virtual uplink system \cite{2004Duality1}. Specifically, for a given analog precoder $ \mat{F} $, the downlink system is illustrated in {\figurename~\ref{downlink}}.  Define the downlink power allocation as $ \vect{p}_{\dl}=\squareBrack{p_{\dl,1},\dots,p_{\dl,K}}^\T $, where $ p_{\dl,i} = \Expect\curlyBrack{ \norm{s_i}_2^2 } $ and $ s_i $ is the data symbol for user $ i $. Then the sum power satisfies that $ \Expect\curlyBrack{ \roundBrack{\vect{FGs}}^{\CT} \vect{FGs} } = \sum_{i=1}^K p_{\dl,i} \squareBrack{ \vect{G}^{\CT}\vect{F}^{\CT}\vect{FG} }_{ii} = \norm{\vect{p}_{\dl}}_1\leq P_{\max}  $, where the last equality holds since $ \mat{G} $ is designed to normalize the columns of $ \mat{FG} $, which can be justified by \eqref{G_DB}. The corresponding virtual uplink model is obtained by switching the role of transmitter and receiver. The data symbol vector $ \vect{s} $ is transmitted from $ K $ independent users through the channel $ \widetilde{\mat{H}}^\CT$. $ \mat{G}^\CT $ now behaves as a normalized multi-user receiver. The quantities $ \mat{G},\widetilde{\mat{H}} $ remain the same as the downlink system. The uplink power allocation $ \vect{p}_{\ul}=\squareBrack{p_{\ul,1},\dots,p_{\ul,K}}^\T $ satisfies the same sum power constraint as the downlink, i.e., $ \norm{\vect{p}_{\ul}}_1\leq P_{\max} $. The duality theory established in \cite{2003DualityOri,2004Duality1} shows that the downlink and virtual uplink can achieve the same data rate region. Moreover, the Pareto-optimal precoder that achieves a boundary point of the data rate region in the downlink is given by the MMSE receiver in the virtual uplink corresponding to the same Pareto-optimal rate point. Motivated by this duality theory, we obtain the digital precoder (called duality-based precoder) from the virtual uplink MMSE receiver.
	
	\begin{figure}[!t]
		\centering
		\subfloat[]{\includegraphics[width=2.8in]{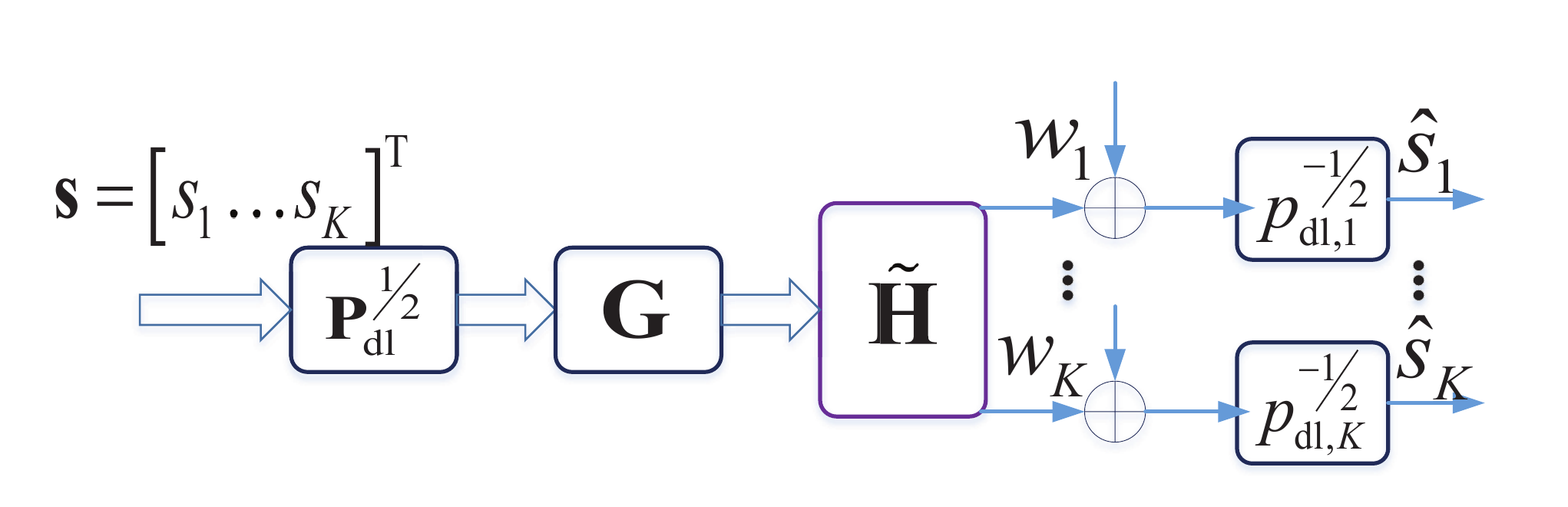}%
			\label{downlink}}
		\hfil
		\subfloat[]{\includegraphics[width=2.8in]{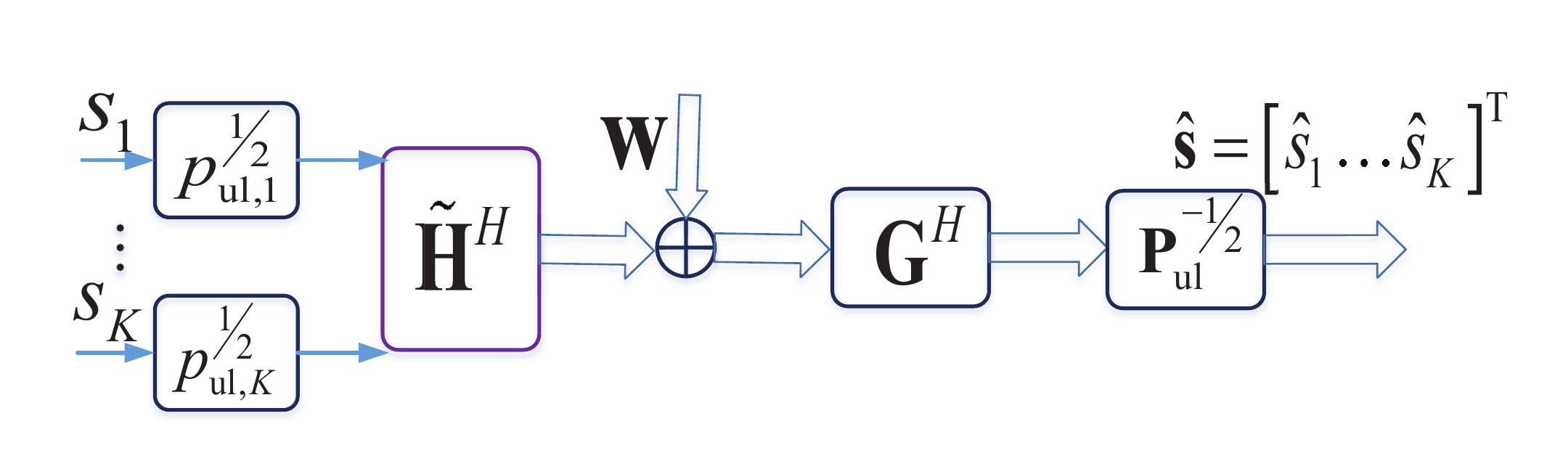}%
			\label{virtual uplink}}
		\caption{An illustration of the downlink and virtual uplink system model. \textbf{(a)} Downlink system model. \textbf{(b)} Virtual uplink system model.}
		\label{DownlinkAndVirtualUplink}
	\end{figure}
	
	In particular, for a given power allocation $ \vect{p} $, the MMSE receiver in the virtual uplink is given by
	\begin{equation}\label{MMSEreceiver}
	\mat{G}^{\mmse}=\roundBrack{\widetilde{\mat{H}}^\CT\mat{P}\widetilde{\mat{H}}+\mat{I}}^{-1}\widetilde{\mat{H}}^\CT\mat{P}.
	\end{equation} 
	where $ \mat{P}=\Diag\roundBrack{\vect{p}} $. Then the baseband digital precoder $ \mat{G} $ in the downlink is given by
	\begin{equation}\label{G_DB}
	\mat{G}=\roundBrack{\widetilde{\mat{H}}^\CT\mat{P}\widetilde{\mat{H}} + \mat{I}}^{-1}\widetilde{\mat{H}}^\CT\mat{P}\mat{\Lambda }^\frac{1}{2},
	\end{equation}  
	where $ \mat{\Lambda }^\frac{1}{2}=\Diag\roundBrack{\squareBrack{\norm{\bar{\vect{g}}_1}^{-1},\dots, \norm{\bar{\vect{g}}_K}^{-1}}} $ is used to normalize the precoding vectors $ \mat{F}\vect{g}_k$'s, and $ \bar{\vect{g}}_k $ is the $ k $-th column of $ \bar{\vect{G}}\triangleq\mat{F}\roundBrack{\widetilde{\mat{H}}^\CT\mat{P}\widetilde{\mat{H}} + \mat{I}}^{-1}\widetilde{\mat{H}}^\CT\mat{P}$.
	
	As a comparison, the RZF digital precoder \cite{2005RZF} is given by 
	\begin{equation}\label{G_RZF}
	\mat{G}_{\RZF}=\widetilde{\mat{H}}^\CT\roundBrack{\widetilde{\mat{H}}\widetilde{\mat{H}}^\CT+\alpha\mat{I}}^{-1} \mat{\Upsilon}^\frac{1}{2},
	\end{equation} 
	where $ \mat{\Upsilon }^\frac{1}{2}=\Diag\roundBrack{\squareBrack{\norm{\tilde{\vect{g}}_1}^{-1},\dots, \norm{\tilde{\vect{g}}_K}^{-1}}} $ is used to normalize the column vectors of $ \mat{F}\mat{G}_{\RZF}$, $ \tilde{\vect{g}}_k $ is the $ k $-th column of $ \tilde{\vect{G}}\triangleq\mat{F}\widetilde{\mat{H}}^\CT\roundBrack{\widetilde{\mat{H}}\widetilde{\mat{H}}^\CT+\alpha\mat{I}}^{-1} $ and $ \alpha $ is the regularization factor.  In \eqref{G_RZF}, the RZF precoder is obtained by assuming that all $ K $ users are scheduled for transmission. A well-known fact is that the RZF precoder requires an explicit user-selection to achieve a good spatial multiplexing gain. The set of scheduled users can essentially be expressed as an indicator function of power allocation \footnote{ Specifically, the power allocated to user $ k $ under the RZF precoder is given by $ p_k\mathbbm{1}\roundBrack{ k \in \mathcal{U} } $, where $ \mathcal{U} $ is the scheduled user set after an explicit user selection. }. Thus, the RZF precoder is not a smooth function of power allocation, making the direct optimization of power allocation intractable. In contrast, the duality-based precoder $ \mat{G}$ in \eqref{G_DB} is obviously a smooth function of power allocation, leading to a tractable power allocation optimization \footnote{ The tractable power allocation optimization in this paper means that we can develop an efficient algorithm to find a ``good'' KKT point without having to deal with the more complicated combinatorial optimization of user selection. The power allocation optimization problem may still be non-convex. }, which can do an implicit user-selection by optimizing the power allocation. Moreover, the complexity of duality-based precoder $ \mat{G} $ in \eqref{G_DB} is similar with that of the RZF precoder $ \mat{G}_{\RZF} $.
	
	Another interesting observation is that the RZF precoder is a special case of the duality-based precoder for a certain choice of power allocation. For example, apply an equal power allocation, i.e., $ p_k=p=\frac{P_{\max}}{K}, \forall k $, to the duality-based precoder $ \mat{G} $ in \eqref{G_DB}. Then from the matrix inverse lemma, it can be verified that $ \mat{G} $ (unnormalized form) degrades to the RZF precoder $ \mat{G}_{\RZF}$ (unnormalized form) with $ \alpha=\frac{1}{p} $ as 
	\begin{equation*}
	\mat{G}=\mat{G}_{\RZF}=\widetilde{\mat{H}}^\CT\roundBrack{\widetilde{\mat{H}}\widetilde{\mat{H}}^\CT+\frac{1}{p}\mat{I}}^{-1}.
	\end{equation*}
	Since $ \alpha=\frac{1}{p} $ is the asymptotic optimal regularization factor for the RZF precoder under the perfect CSI \cite{2005RZF,2012LargeSystemLinearPrecoding}, this further indicates that the performance of duality-based precoder after the power allocation optimization will be better than that of the RZF precoder.
	
	Considering the impact of effective CSI estimation error, the final digital precoder $ \mat{G} $ is given by
	\begin{equation}
	\mat{G}=\roundBrack{\hat{\widetilde{\mat{H}}}^\CT\mat{P}\hat{\widetilde{\mat{H}}} + \mat{I}}^{-1}\hat{\widetilde{\mat{H}}}^\CT\mat{P}\mat{\Lambda }^\frac{1}{2}.
	\end{equation}  
	where $ \hat{\widetilde{\mat{H}}}=\squareBrack{\hat{\widetilde{\vect{h}}}_1,\dots,\hat{\widetilde{\vect{h}}}_K}^\CT $ is the estimated effective CSI matrix and $ \mat{\Lambda }^\frac{1}{2} $ is the corresponding normalization matrix.

	\subsection{Achievable Data Rate}
	In the proposed RCSHP, we consider the randomized control policy as elaborated in section \ROMAN{2}-B. Under a given realization of control state $ l $, the control variable is given by $ \vect{\Gamma}\roundBrack{l}=\squareBrack{\pmb{\theta}\roundBrack{l}^\T, \vect{p}\roundBrack{l}^\T}^\T $, where $ \pmb{\theta}\roundBrack{l}$ and $\vect{p}\roundBrack{l} $ are the phase shifting vector of the analog precoder and the power allocation vector at the control state $ l $, respectively. For a given channel realization $ \mat{H} $ and channel estimation noise realization $ \vect{N}=\squareBrack{\vect{n}_1,\dots,\vect{n}_K}^{\CT} $, the instantaneous achievable data rate of the user $ k $ at the control state $ l $ is 
	\begin{equation}
	r_k\roundBrack{\vect{\Gamma}\roundBrack{l}; \mat{H},\mat{N},l}=\log\roundBrack{1 + \dfrac{p_k\abs{\vect{h}_k^\CT \mat{F}\vect{g}_k}^2}{\sum_{i\neq k} p_i\abs{\vect{h}_k^\CT\mat{F}\vect{g}_i}^2 + 1 }},
	\end{equation}
	where the control state $ l $ is dropped off in the specific expression of $ r_k $ for conciseness. Note that $ \mat{F} $ is a function of  $ \pmb{\theta}\roundBrack{l} $ and $ \mat{G} $ is a function of $ \pmb{\theta}\roundBrack{l},\vect{p}\roundBrack{l},\mat{H},\vect{N} $. Thus we can explicitly express $ r_k $ as a function of $ \vect{\Gamma}\roundBrack{l}$ which depends on the random states $ \mat{H},\vect{N},l $. Then for a given control policy $ \vect{\Omega}=\curlyBrack{\vect{\Gamma},\vect{q}} $, the average achievable data rate of user $ k $ is given by 
	\begin{equation}
	\bar{r}_k=\sum_{l=1}^{L}q_l\Expect\nolimits_{\mat{H},\vect{N}}\squareBrack{r_k\roundBrack{\vect{\Gamma}\roundBrack{l};\mat{H},\vect{N},l}}.
	\end{equation}  
	Define $ \vect{\bar{r}}\triangleq\squareBrack{\bar{r}_1,\dots,\bar{r}_K}^\T $ as the average data rate vector.
	
	\section{Problem Formulation}
	The joint optimization of the randomized control policy, i.e., the time-sharing factors and the associated analog precoders and power allocations, can be formulated as the following general utility maximization problem:
    	\begin{equation}\label{OriginalProblem}
    	\mathcal{P}: \quad \max_{\vect{\Gamma}\in\mathcal{G},\vect{q}\in\mathcal{Q}}\quad f\roundBrack{ \vect{\Gamma},\vect{q}}\triangleq U\roundBrack{\vect{\bar{r}}}
    \end{equation} 
	where $ U\roundBrack{\cdot} $ is a general utility function, which is assumed to be continuously differentiable, concave and nondecreasing for all $ \vect{\bar{r}}\geq 0 $, and the gradient of $ U\roundBrack{\vect{\bar{r}}} $ with respect to (w.r.t.) $ \vect{\bar{r}} $ is Lipschitz continuous. This general utility function $ U\roundBrack{\vect{\bar{r}}} $ includes many important network utilities as special cases, such as $ \alpha $-fair utility \cite{2000AlphaFairness}, sum rate ($ U\roundBrack{\vect{\bar{r}}}=\sum_{k=1}^{K}\bar{r}_k $, a special case of $ \alpha $-fair when $ \alpha=0 $ ) and proportional fairness utility ($ U\roundBrack{\vect{\bar{r}}}=\sum_{k=1}^{K}\log\roundBrack{\bar{r}_k+\epsilon}, $ where $ \epsilon>0 $ is a small number to avoid the singularity at $ \bar{r}_k=0 $, also a special case of $ \alpha $-fair when $ \alpha=1 $).
	
    In our design, the randomized control policy is assumed to be adaptive to the channel statistics, since we consider a practical scenario where the pilot resource is limited. In this case, it is infeasible to obtain the instantaneous CSI and thus it is more practical to adapt the randomized control policy according to the channel statistics. At the beginning of each coherence time of channel statistics, we firstly generate an appropriate number of channel samples and channel estimation noise samples according to the statistical CSI (channel distribution) $ \mathcal{H} $, e.g., $ \vect{h}_k\sim\CN{\vect{0},\mat{C}_k}, \forall k $, for Gaussian channel distribution, and the channel estimation noise distribution $ \vect{n}_k\sim\CN{\vect{0},\mat{I}}, \forall k $, to solve the problem $ \mathcal{P} $ in order to obtain the optimized randomized control policy $ \vect{\Omega}^{\star}=\curlyBrack{\vect{\Gamma}^{\star}, \vect{q}^{\star}}  $. 
	
	During the maximization of utility, the number of the selected beams of each user's effective channel in each compatible user group tends to be smaller than the number of assigned pilot symbols $ T_p $, so that the channel sparsifying is implicitly realized by jointly optimizing the time-sharing factors and the associated control variables, leading to an improved effective CSI estimation quality. Note that there is no need to explicitly carry out the user grouping as it is automatically realized by the optimization of control variables.  Each optimized control variable corresponds to a compatible group of users, such that the effective CSI in each compatible user group not only has enough spatial DoF to support the simultaneous transmission to these users, but also is sparse enough to achieve a good effective CSI estimation quality by the limited number of pilots. Each optimized time-sharing factor corresponds to the proportion of time that serves one compatible user group. As such, the \textit{active channel sparsification} can be implicitly realized by maximizing the utility function over the randomized control policy, and thus the expected long-term utility can be achieved.
	
    After the optimized randomized control policy $ \vect{\Omega}^{\star}  $ has been calculated, we simply apply it at each time slot during the current coherence time of channel statistics, and the digital precoder $ \mat{G} $ is adaptive to the instantaneous effective CSI to support the downlink transmission. A toy example has been elaborated in section \ROMAN{2}-B and the illustrations are shown in {\figurename~\ref{user_group_example}} and {\figurename~\ref{OpratingProcess}}.
	
	However, there are several challenges in finding KKT solutions of the problem $ \mathcal{P} $, as elaborated below. First, the objective function is nether convex nor concave, and it contains an expectation operator,  which causes that the objective function usually does not have a closed-form expression. Moreover, there are three random system states $ \mat{H}, \mat{N} $ and $ l $, and the probability measure of the control state $ l $ depends on the time-sharing vector $ \vect{q} $. To address these challenges, we propose an efficient SSCA-RCSHP algorithm to find KKT points \footnote{ In practice, the probability of converging to a non-locally optimal KKT point (e.g., saddle point) is very small because these KKT points are not stable. Therefore, the KKT points found by the algorithm are usually locally optimal.  } of the problem $ \mathcal{P} $.
	
	\section{Algorithm Design}
	In this section, we propose an efficient stochastic successive convex approximation algorithm called SSCA-RCSHP to solve the problem $ \mathcal{P} $. Algorithm \ref{SSCA} summarizes the key steps of the SSCA-RCSHP. At each iteration $ t $, the randomized control policy $ \mat{\Omega}=\curlyBrack{\vect{\Gamma}, \vect{q}} $ is updated by solving a convex surrogate problem obtained by replacing the objective function $ f\roundBrack{\vect{\Gamma}, \vect{q}} $ with its convex surrogate function $ \bar{f}^t\roundBrack{\vect{\Gamma}, \vect{q}} $. 
	
	Specifically, at the $ t $-th iteration,  $ T_{HN} $ channel and channel estimation noise realizations $ \curlyBrack{\mat{H}^{t}\roundBrack{i},\mat{N}^{t}\roundBrack{i}}_{i=1,\dots,T_{HN}} $ are firstly generated according to the statistical CSI $ \mathcal{H} $ in Step 1. Then the surrogate function $ \bar{f}^t\roundBrack{\vect{\Gamma},\vect{q}} $ is updated based on $ \curlyBrack{\mat{H}^{t}\roundBrack{i},\mat{N}^{t}\roundBrack{i}}_{i=1,\dots,T_{HN}} $ and the current iterate $ \vect{\Gamma}^{t}, \vect{q}^t $ in Step 2 as 
	\begin{equation}\label{SurrogateFunc}
	\bar{f}^t\roundBrack{\vect{\Gamma}, \vect{q}}= U\roundBrack{  \vect{\hat{\bar{r}}}^t\roundBrack{\vect{q}}   }-\tau_{q}\norm{\vect{q}-\vect{q}^t}_2^2
	+\roundBrack{\mathbf{f}^t_{\vect{\Gamma}}}^\T\roundBrack{\vect{\Gamma}-\vect{\Gamma}^{t}}-\tau_{\Gamma}\norm{\vect{\Gamma}-\vect{\Gamma}^{t}}_2^2,
	\end{equation}
	where $ \tau_{q},\tau_{\Gamma}>0 $ are two constants,  $ \vect{\hat{\bar{r}}}^t\roundBrack{\vect{q}}  =\squareBrack{\sum_{l=1}^{L}q_l\hat{r}_1^t\roundBrack{l},\dots,\sum_{l=1}^{L}q_l\hat{r}_K^t\roundBrack{l}}^\T $ is an approximation of the average data rate vector and $ \hat{r}_k^t\roundBrack{l} $ is the approximate conditional average data rate of the user $ k $ under the $ l $-th analog precoding and power allocation state, which is recursively updated as
    \begin{equation}\label{FuncValueUpdate}
	\hat{r}^t_{k}\roundBrack{l}=\roundBrack{1-\rho_t}\hat{r}^{t-1}_{k}\roundBrack{l}+
	\rho_t\sum_{i=1}^{T_{HN}}\dfrac{r_k\roundBrack{ \vect{\Gamma}^t\roundBrack{l};\mat{H}^t\roundBrack{i},\vect{N}^t\roundBrack{i}}}{T_{HN}}, \forall k, \forall l,
	\end{equation}
	with $ \hat{r}^{-1}_{k}=0, \forall k,\forall l $.  $ \mathbf{f}^t_{\vect{\Gamma}} $ is an approximation of the partial derivative $ \nabla_{\vect{\Gamma}} U\roundBrack{\vect{\bar{r}}} $, which is updated recursively as 
	\begin{equation}\label{GradientUpdate}
	\mathbf{f}^t_{\vect{\Gamma}}=\roundBrack{1-\rho_t}\mathbf{f}^{t-1}_{\vect{\Gamma}}
	+\rho_t\sum_{i=1}^{T_{HN}}\dfrac{\mathbf{J}_{\vect{\Gamma}}\roundBrack{\vect{\Gamma}^t,\vect{q}^t; \mat{H}^t\roundBrack{i},\mat{N}^t\roundBrack{i}}\nabla_{\vect{\bar{r}}}U\roundBrack{\vect{\hat{\bar{r}}}^t\roundBrack{\vect{q}^t}} }{T_{HN}},
	\end{equation}
	with $ \mathbf{f}^{-1}_{\vect{\Gamma}}=\vect{0} $, where $ \mathbf{f}^t_{\vect{\Gamma}}=\squareBrack{\roundBrack{\mathbf{f}^t_{\vect{\Gamma}\roundBrack{1}}}^\T,\dots,\roundBrack{\mathbf{f}^t_{\vect{\Gamma}\roundBrack{l}}}^\T,\dots,\roundBrack{\mathbf{f}^t_{\vect{\Gamma}\roundBrack{L}}}^\T}^\T $, $ \rho_t\in\left(0,1\right] $ is a sequence to be properly chosen, 
	$ \mathbf{J}_{\vect{\Gamma}}\roundBrack{\vect{\Gamma},\vect{q}; \mat{H},\mat{N}} $ is the Jacobian matrix of the data rate vector $ \tilde{\vect{r}}\roundBrack{\vect{\Gamma},\vect{q}; \mat{H},\mat{N}}=\squareBrack{\sum_{l=1}^{L}q_l r_1\roundBrack{l},\dots,\sum_{l=1}^{L}q_l r_K\roundBrack{l}}^\T $ w.r.t. $ \vect{\Gamma} $ and its detailed expression is given by Appendix A. Note that $ U\roundBrack{ \vect{\hat{\bar{r}}}^t\roundBrack{\vect{q}} } $ is a concave function over $ \vect{q} $, since $ \vect{\hat{\bar{r}}}^t\roundBrack{\vect{q}} $ is a linear function w.r.t. $ \vect{q} $, $ U\roundBrack{\cdot} $ is assumed to be a concave function and linear mapping preserves concavity of functions \cite{2004CVXBook}.
	
	 We intuitively explain the need for each of the terms in \eqref{SurrogateFunc}. The reason we directly capture the dependence on $ \vect{q} $ by the utility function $ U\roundBrack{\cdot} $ is that the original problem \eqref{OriginalProblem} is a convex problem w.r.t. $ \vect{q} $ if $ \vect{\Gamma} $ is given. Thus at each iteration, we can directly solve a convex subproblem over $ \vect{q} $ without any convex approximation. However, the original problem \eqref{OriginalProblem} is still a non-convex problem w.r.t. $ \vect{\Gamma} $ even if $ \vect{q} $ is given. As such, at each iteration, we employ a local gradient $ \mathbf{f}^t_{\vect{\Gamma}} $, which is an approximation of the partial derivative $ \nabla_{\vect{\Gamma}} U\roundBrack{\vect{\bar{r}}} $, to first-order approximate the objective function $ U\roundBrack{\vect{\bar{r}} } $ over $ \vect{\Gamma} $, such that the corresponding subproblem is computational tractable. Notice that we drop off the constant term for conciseness. In addition, the two quadratic terms are introduced to guarantee that the surrogate problem is strongly convex, thereby augmenting the convergence stability.     
	
	After updating the surrogate function, the following surrogate problem is solved:
    \begin{equation}\label{SurrogateProblem}
    	\widehat{\mathcal{P}}: \quad \max_{\vect{\Gamma}\in\mathcal{G},\vect{q}\in\mathcal{Q}}\quad  \bar{f}^t\roundBrack{\vect{\Gamma}, \vect{q}}
    \end{equation} 
	Then the control policy is updated based on the solution of $ \widehat{\mathcal{P}} $, as summarized in Step 3 and 4.
	
	Specifically, the problem $ \widehat{\mathcal{P}} $ can be decomposed into $ L+1 $ convex subproblems w.r.t. $ \vect{q} $ and $ \vect{\Gamma}\roundBrack{l},\forall l $, respectively. Thus in Step 3a, the optimal solution $ \bar{\vect{q}}^t $ for $ \widehat{\mathcal{P}} $ is obtained by solving the following subproblem:
	\begin{equation}\label{Obtain_q}
	\mathcal{P}_q: \quad \bar{\vect{q}}^t= \arg\max_{\vect{q}\in\mathcal{Q}}U\roundBrack{
	 \vect{\hat{\bar{r}}}^t\roundBrack{\vect{q}} }-\tau_{q}\norm{\vect{q}-\vect{q}^t}_2^2. 
	\end{equation} 
    This subproblem is convex as explained above and thus can be solved by standard convex optimization methods \cite{2004CVXBook}. Then the time-sharing vector $ \vect{q} $ is updated in Step 3b as 
	\begin{equation}\label{Update_q}
	\vect{q}^{t+1}=\roundBrack{1-\gamma_t}\vect{q}^t+\gamma_t\bar{\vect{q}}^t,
	\end{equation}
	where $\gamma_t\in\left(0,1\right] $ is a sequence to be properly chosen. In Step 4a , the optimal solution $\bar{\vect{\Gamma}}^t=\squareBrack{\roundBrack{\bar{\vect{\Gamma}}^t\roundBrack{1}}^\T,\dots,\roundBrack{\bar{\vect{\Gamma}}^t\roundBrack{L}}^\T}^\T $ for $ \widehat{\mathcal{P}} $ is obtained by independently solving the following $ L $ subproblems:
	\begin{equation}\label{Obtain_Gamma}
	\mathcal{P}_{\Gamma_l}:\quad\bar{\vect{\Gamma}}^t\roundBrack{l}=\arg\max_{\vect{\Gamma}\roundBrack{l}\in\mathcal{G}_l} \roundBrack{\mathbf{f}^t_{\vect{\Gamma}\roundBrack{l}}}^\T\roundBrack{\vect{\Gamma}\roundBrack{l}-\vect{\Gamma}^{t}\roundBrack{l}}
	-\tau_{\Gamma}\norm{\vect{\Gamma}\roundBrack{l}-\vect{\Gamma}^{t}\roundBrack{l}}_2^2,
	\end{equation} 
	for $ l=1,\dots,L $, where  $ \mathcal{G}_l=\curlyBrack{\pmb{\theta}\roundBrack{l}\in\squareBrack{0,2\pi}^{MS}, \vect{p}\roundBrack{l}\in\R^K_+,  \vect{1}^\T\vect{p}\roundBrack{l}\leq P_{\max}  } $.
	
	  Problem $ \mathcal{P}_{\Gamma_l} $ is a convex quadratic problem, which has a closed-form solution $ 	\bar{\vect{\Gamma}}^t\roundBrack{l} = \mathbb{P}_{\mathcal{G}_l}\squareBrack{\vect{\Gamma}^t\roundBrack{l}+\frac{ \mathbf{f}^t_{\vect{\Gamma}\roundBrack{l}} }{2\tau_{\Gamma}}} $, where $ \mathbb{P}_{\mathcal{G}_l}\squareBrack{\cdot} $ denotes the projection onto the convex set $ \mathcal{G}_l $.   Subsequently, the aggregated vector of control variables $ \vect{\Gamma} $ is updated in Step 4b according to
	\begin{equation}\label{Update_Gamma}
	\vect{\Gamma}^{t+1}=\roundBrack{1-\gamma_t}\vect{\Gamma}^{t}+\gamma_t\bar{\vect{\Gamma}}^t,
	\end{equation}
	Then the above steps (Step 1 to Step 4) are carried out until convergence.
	
	\begin{algorithm}[!t] 
		\caption{SSCA-RCSHP Algorithm}
		\label{SSCA}
		\renewcommand{\algorithmicrequire}{\textbf{Initialize:}}
		\begin{algorithmic}
			\REQUIRE $ \vect{\Gamma}^0, \vect{q}^0, \hat{r}^{-1}_k=0,\forall k, \mathbf{f}^{-1}_{\vect{\Gamma}}=0, T_{HN}, t=0.  $
			\STATE \textbf{Step 1:} Generate $ T_{HN} $ new channel and channel estimation noise realizations $ \curlyBrack{\mat{H}^{t}\roundBrack{i},\mat{N}^{t}\roundBrack{i}}_{i=1,\dots,T_{HN}} $ according to the statistical CSI $ \mathcal{H} $ at the $  t $-th iteration.
			\STATE \textbf{Step 2:} Update the surrogate function by \eqref{SurrogateFunc}.
			\STATE \textbf{Step 3a:} Solve \eqref{Obtain_q} to obtain the optimal solution $ \bar{\vect{q}}^t $.
			\STATE \textbf{Step 3b:} Update $ \vect{q}^{t+1} $ according to \eqref{Update_q}.
			\STATE \textbf{Step 4a:} Distributedly solve $ L $ subproblems \eqref{Obtain_Gamma} to obtain the optimal solution $ \bar{\vect{\Gamma}}^t $.
			\STATE \textbf{Step 4b:} Update $ \vect{\Gamma}^{t+1} $ according to \eqref{Update_Gamma}.
			\STATE \textbf{Step 5:} Let $ t=t+1 $ and return to Step 1.
		\end{algorithmic}
	\end{algorithm}
	
	
	\section{Convergence and Complexity Analysis}
	\subsection{Convergence Analysis}
	We establish the convergence of SSCA-RCSHP to KKT solutions. Notice that the limiting point is obtained by averaging over all the previous solutions from the surrogate problem \eqref{SurrogateProblem}, which makes it difficult to show that the limiting point is a KKT point of the original Problem \eqref{OriginalProblem}. To address this challenge, we need to make some assumptions on the sequence of parameters $ \curlyBrack{\rho_t} $ and $ \curlyBrack{\gamma_t} $.
	
	\begin{assumption}
		\textit{(Assumptions on  $ \curlyBrack{\rho_t} $, $ \curlyBrack{\gamma_t} $):} 
		\begin{itemize}
			\item[1)] $ \rho_t\rightarrow 0 $,  $ \frac{1}{\rho_t}\leq O\roundBrack{t^{\beta}} $ for some $ \beta\in\roundBrack{0,1} $, $ \sum_t\roundBrack{\rho_t}^2<\infty $, $ \sum_t\rho_t t^{-\frac{1}{2}}<\infty $.
			\item[2)] $ \gamma_t\rightarrow 0 $, $ \sum_t\gamma_t=\infty  $, $ \sum_t\roundBrack{\gamma_t}^2<\infty $
			\item[3)] $\lim_{t\rightarrow\infty}\frac{\gamma_t}{\rho_t}=0 $
		\end{itemize}
	\end{assumption}
	
	Note that the condition $ \frac{1}{\rho_t}\leq O\roundBrack{t^{\beta}} $ for some $ \beta\in\roundBrack{0,1} $ is almost the same as $ \sum_t \rho_t = \infty $, which is a common assumption in stochastic optimization algorithms \cite{2016ParallelDecomposition}. With Assumption 1, we first prove a key lemma that will support the final convergence. The following lemma proves the convergence of the surrogate objective function.
	
	\begin{lemma}
		\textit{(Convergence of the surrogate objective function): Suppose Assumption 1 is satisfied. Consider a subsequence $ \curlyBrack{\vect{\Gamma}^{t_j}, \vect{q}^{t_j}}_{j=1}^{\infty} $ converging to a limiting point $ \curlyBrack{\vect{\Gamma}^{*}, \vect{q}^{*}} $, and define a function
			\begin{equation*}
			\begin{split}
			\hat{f}\roundBrack{\vect{\Gamma}, \vect{q}}&\triangleq U\roundBrack{\vect{\bar{r}}\roundBrack{\vect{\Gamma}^{*},\vect{q}}}-\tau_{q}\norm{\vect{q}-\vect{q}^{*}}_2^2\\
			&+\nabla_{\vect{\Gamma}}^{\T}f\roundBrack{\vect{\Gamma}^{*}, \vect{q}^{*}}\roundBrack{\vect{\Gamma}-\vect{\Gamma}^{*}}-\tau_{\Gamma}\norm{\vect{\Gamma}-\vect{\Gamma}^{*}}_2^2,
			\end{split}
			\end{equation*}
			which satisfies $ \hat{f}\roundBrack{\vect{\Gamma}^{*}, \vect{q}^{*}}=f\roundBrack{\vect{\Gamma}^{*}, \vect{q}^{*}} $, $ \nabla_{\vect{\Gamma}}\hat{f}\roundBrack{\vect{\Gamma}^{*}, \vect{q}^{*}}= \nabla_{\vect{\Gamma}}f\roundBrack{\vect{\Gamma}^{*}, \vect{q}^{*}} $ and $ \nabla_{\vect{q}}\hat{f}\roundBrack{\vect{\Gamma}^{*}, \vect{q}^{*}}= \nabla_{\vect{q}}f\roundBrack{\vect{\Gamma}^{*}, \vect{q}^{*}} $. Then almost surely, we have
			\begin{equation*}
			\lim_{j\rightarrow\infty}\bar{f}^{t_j}\roundBrack{\vect{\Gamma}, \vect{q}}=\hat{f}\roundBrack{\vect{\Gamma}, \vect{q}}, \forall \vect{q}\in\mathcal{Q}, \forall\vect{\Gamma}\in\mathcal{G},
			\end{equation*} }
	\end{lemma}
	Please refer to Appendix B for the proof. With Lemma 1, the following convergence theorem can be
	proved.
	
	\begin{theorem}
		\textit{	(Convergence of Algorithm \ref{SSCA}): Suppose Assumption 1 is satisfied. For any subsequence $ \curlyBrack{\vect{\Gamma}^{t_j}, \vect{q}^{t_j}}_{j=1}^{\infty} $ converging to a limiting point $ \curlyBrack{\vect{\Gamma}^{*}, \vect{q}^{*}} $,  $ \curlyBrack{\vect{\Gamma}^{*}, \vect{q}^{*}} $ is a KKT point of Problem \eqref{OriginalProblem} almost surely.}
	\end{theorem}
	Please refer to Appendix C for the proof.
	
	\subsection{Complexity Analysis}
	We analyze the computational complexity of the proposed RCSHP scheme, which is dominated by computing the $ \mathbf{f}^t_{\vect{\Gamma}} $ and solving $ (L+1) $ convex problems ( problem \eqref{Obtain_q} and \eqref{Obtain_Gamma} ) at each iteration of the proposed SSCA-RCSHP algorithm.  The per-iteration computational complexity of the former is $ C_f = \max\curlyBrack{ O\roundBrack{ L M^3 S^3 K^2 } , O\roundBrack{ L M^2 S^3 K^3 T_p } } + O\roundBrack{L M^3 S^2 K T_{HN} } $ and the latter is $ C_s = O\roundBrack{L} + O\roundBrack{ L\roundBrack{MS+K} }  $. Therefore, the per-iteration computational complexity of the proposed method is $ C_t = C_f + C_s$. Note that the computational complexity we provide is the worst-case situation. Actually, the computation of $ \mathbf{f}^t_{\vect{\Gamma}} $ includes many sparse matrix operations, such as the calculations of $ \bar{ \mat{A} } $ and $ \bar{ \mat{D} }_h $ in the Appendix A, which can be significantly accelerated in practice.

	\section{Simulation Results}
	We adopt the active channel sparsification (ACS) scheme in \cite{2018ActiveChannelSpar} and the two-timescale hybrid precoding (THP) scheme in \cite{2018SSCA-THP} as baselines to compare with the proposed RCSHP. As in \cite{2018SSCA-THP} and \cite{2018ActiveChannelSpar}, we adopt a geometry-based channel model and a COST 2100 channel model, both with a half-wavelength spaced uniform linear array (ULA) for simulations. For the geometry-based channel, the channel vector of user $ k $ can be expressed as $ \vect{h}_k=\sum_{i=1}^{N_p}\alpha_{k,i}\vect{a}\roundBrack{\varphi_{k,i}} $,  where $ \vect{a}\roundBrack{\varphi}=\squareBrack{1, e^{j\pi\sin\roundBrack{\varphi}},\dots, e^{j\roundBrack{M-1}\pi\sin\roundBrack{\varphi}}}^{\T} $ is the array response vector, $ N_p $ is the number of scattering paths, $ \varphi_{k,i} $'s are angles of departure (AoD) which are Laplacian distributed with an angular spread $ \sigma_{AS}=10 $ and $ \alpha_{k,i}\sim\CN{0,\sigma^2_{k,i}} $, $ \sigma^2_{k,i} $'s are randomly generated from an exponential distribution and normalized such that $ \sum_{i=1}^{N_p}\sigma^2_{k,i}=g_k $, and $ g_k $ represents the path gain of user $ k  $. The path gains $ g_k $'s are uniformly generated between -10 dB and 10 dB and the number of scattering paths for each user is  $ N_p=8 $. For the COST 2100 channel, same with \cite{2018ActiveChannelSpar}, we consider three scattering clusters which are randomly located within the angular range $ [-1,1) $ (parameterized by $ \xi = \frac{ \sin \theta }{ \sin \theta_{\max} } $, where $ \theta \in [-\theta_{\max}, \theta_{\max} ) $ is the AoD). The size of the angular interval for each scattering cluster is 0.2. The channel angular scattering function (ASF) for each user is obtained by randomly selecting two from these three clusters. The ranks of users' channel covariance matrices generated from the geometry-based and COST 2100 channel in this setting are 8 and 36, respectively. This indicates that in our simulations the geometry-based channel is relatively sparse and the COST 2100 channel is relatively rich-scattering.
	
	Consider $ M = 64 $ antennas and $ S = 8 $ RF chains. We provide the sum rate and PFS utility as two main performance metrics, and the corresponding number of served users is $ K=8 $ and $ K=12 $, respectively. In order to compare with the performance of the original version of ACS in \cite{2018ActiveChannelSpar}, we also plot the performance of RCSHP and ACS with a full set of RF chains, denoted by `RCSHP-free' and `ACS-free', respectively.  Define the signal-to-noise ratio $ SNR = \frac{ P_{\max} }{N_0} $, where $ P_{\max} $ is the transmit power budget at the BS and $ N_0 $ is the noise power in the channel estimation phase. Note that $ N_0 $ is normalized to one for simplicity in this paper. $ SNR $ is $ 10 $ dB for all simulations unless otherwise stated. The number of time-sharing factors is $ L = 4 $, each time slot contains $ T=20 $ symbols and each coherence time of channel statistics consists of 200 time slots. For both RCSHP and THP, the number of channel samples generated at each iteration is $ T_{HN} = 9 $ and the maximum number of iterations is 100. These two parameters can be flexibly adjusted according to the practical needs.
	
		\begin{figure}[!t]
		\centering
		\includegraphics[width=3in]{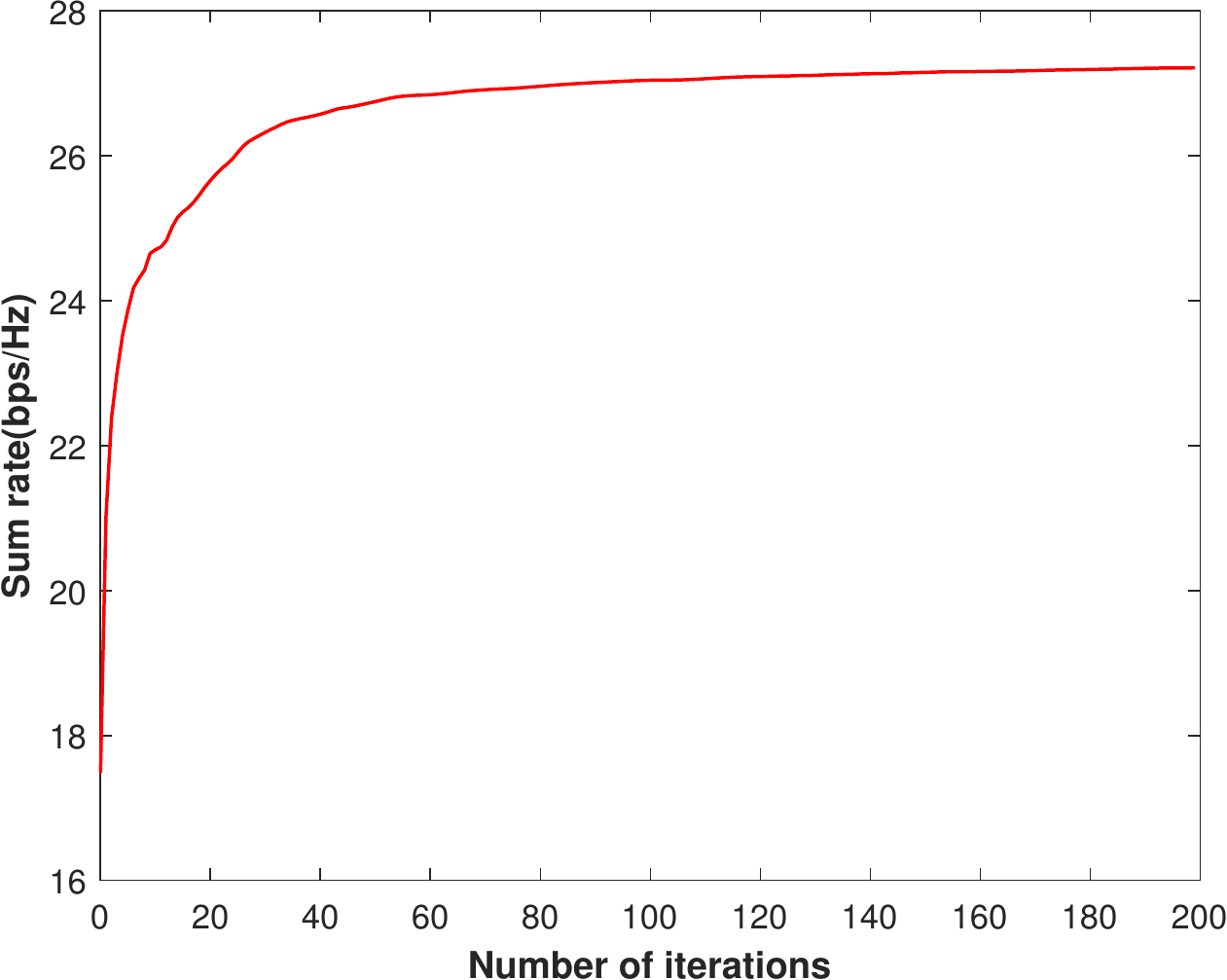}
		\caption{Convergence of SSCA-RCSHP.}
		\label{Convergence}
	\end{figure}
	\subsection{Convergence of the proposed SSCA-RCSHP}
	We use the sum rate in the geometry-based channel environment to illustrate the convergence of SSCA-RCSHP. The number of pilots is $ T_p = 8 $. We set the number of samples generated at each iteration and the number of iterations to be 20 and 200, respectively. It can be seen in \figurename{~\ref{Convergence}} that the proposed SSCA-RCSHP can converge to a KKT point within about 50 iterations.

	\subsection{Sum Rate Maximization}
	\subsubsection{Sum rate versus the number of pilots}
	 According to \figurename{~\ref{SumRate_vs_Pilots_Caire}} and \figurename{~\ref{SumRate_vs_Pilots_Chang}}, we can observe that the proposed RCSHP achieves a larger sum rate than other schemes in both channel environments, and can also attain a good performance with a smaller number of pilots. This indicates that the RCSHP design is robust w.r.t. both non-sparse and sparse channel environments and enables to more ``efficiently'' realize the \textit{active channel sparsification} as well. 
	\begin{figure}[!t]
		\centering
		\subfloat[]{\includegraphics[width=2.5in]{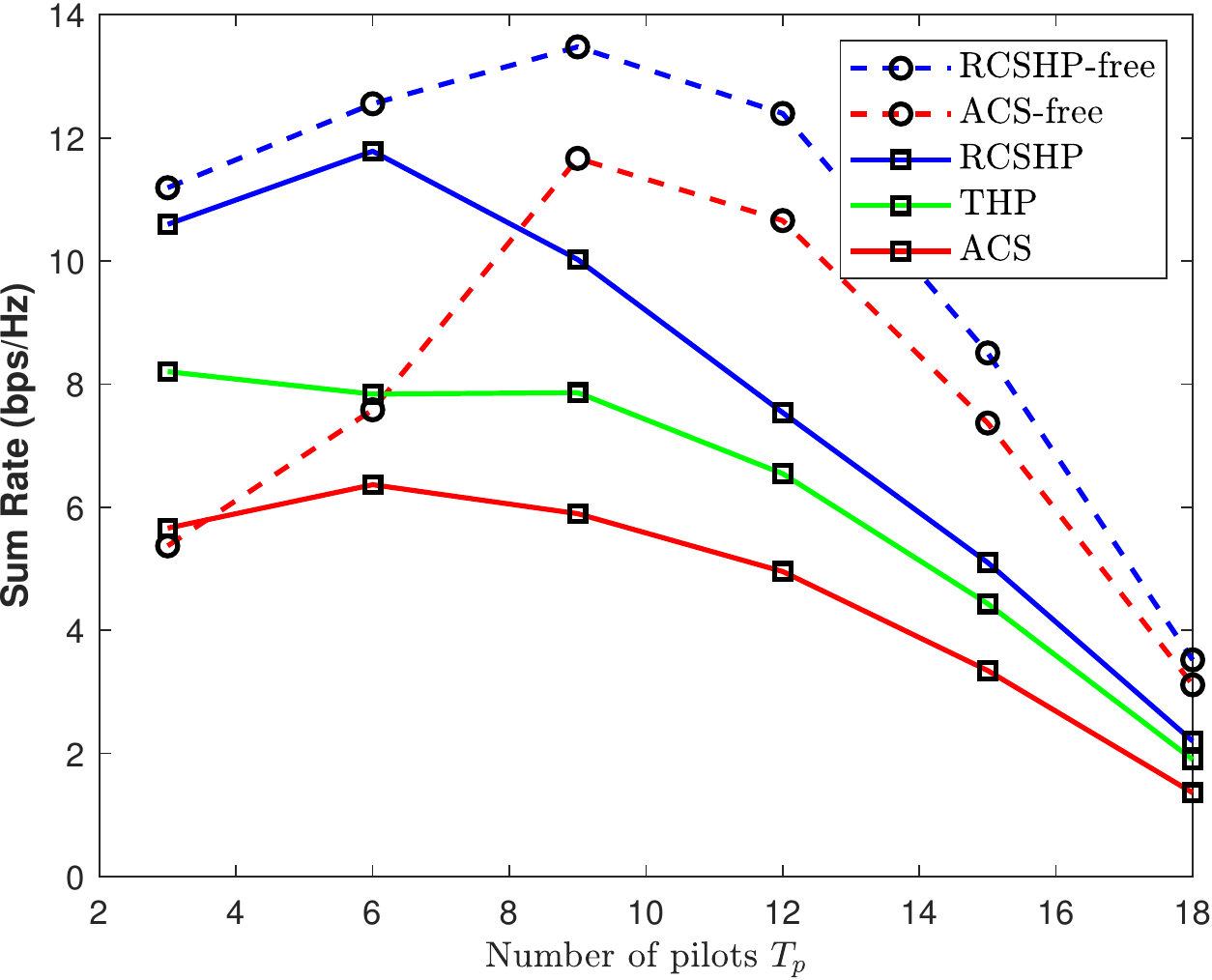}%
			\label{SumRate_vs_Pilots_Caire} }
		\hfil
		\subfloat[]{\includegraphics[width=2.5in]{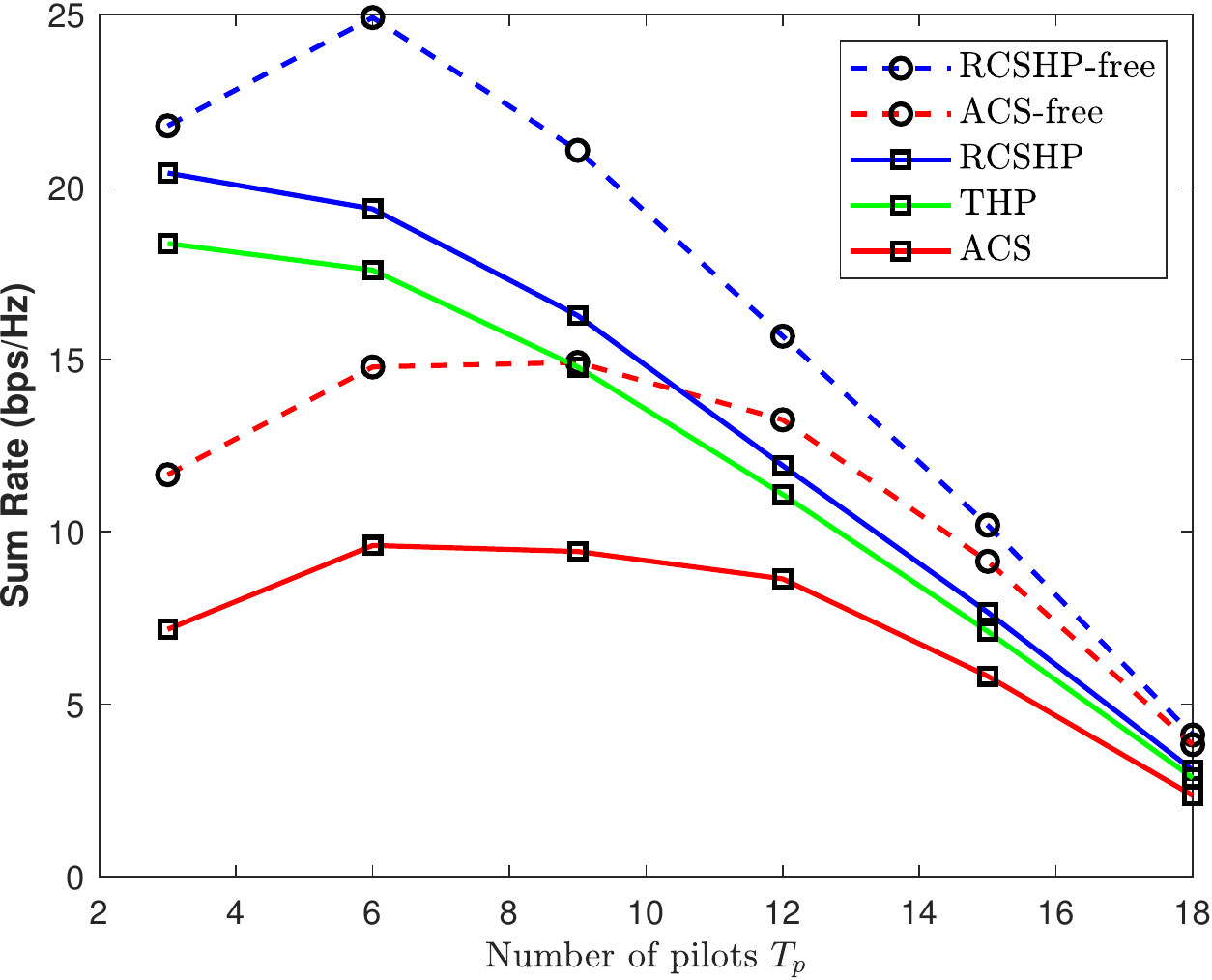}%
			\label{SumRate_vs_Pilots_Chang} }
		\caption{ \textbf{(a)} Sum rate versus the number of pilots in the COST 2100 channel. \textbf{(b)} Sum rate versus the number of pilots in the geometry-based channel.}
	\end{figure}
	
	\subsubsection{Sum rate versus the $ SNR $}
	The number of pilot symbols is chosen to be a moderate size $ T_p = 9 $. It can be observed in \figurename{~\ref{SumRate_vs_Pilots_Caire_SNR}} and \figurename{~\ref{SumRate_vs_Pilots_Chang_SNR}} that the sum rate of the proposed RCSHP scheme increases (almost) linearly as the $ SNR $ grows and RCSHP achieves a better performance than those of other precoding schemes, which indicates that the RCSHP design enables to realize a better beam (angular direction)-user selection, such that the transmission performance can be significantly improved.
    	\begin{figure}[!t]
    	\centering
    	\subfloat[]{\includegraphics[width=2.5in]{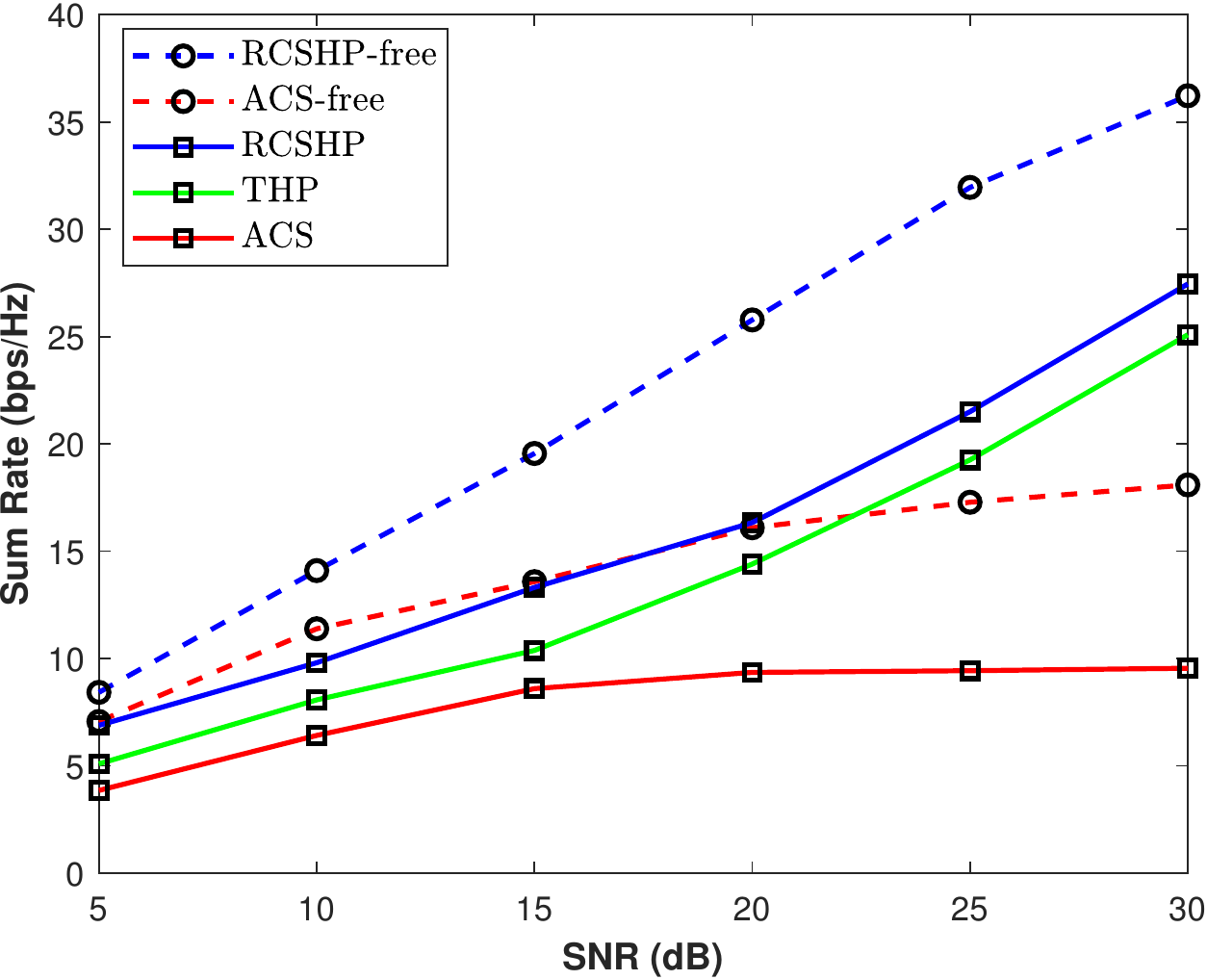}%
    		\label{SumRate_vs_Pilots_Caire_SNR} }
    	\hfil
    	\subfloat[]{\includegraphics[width=2.5in]{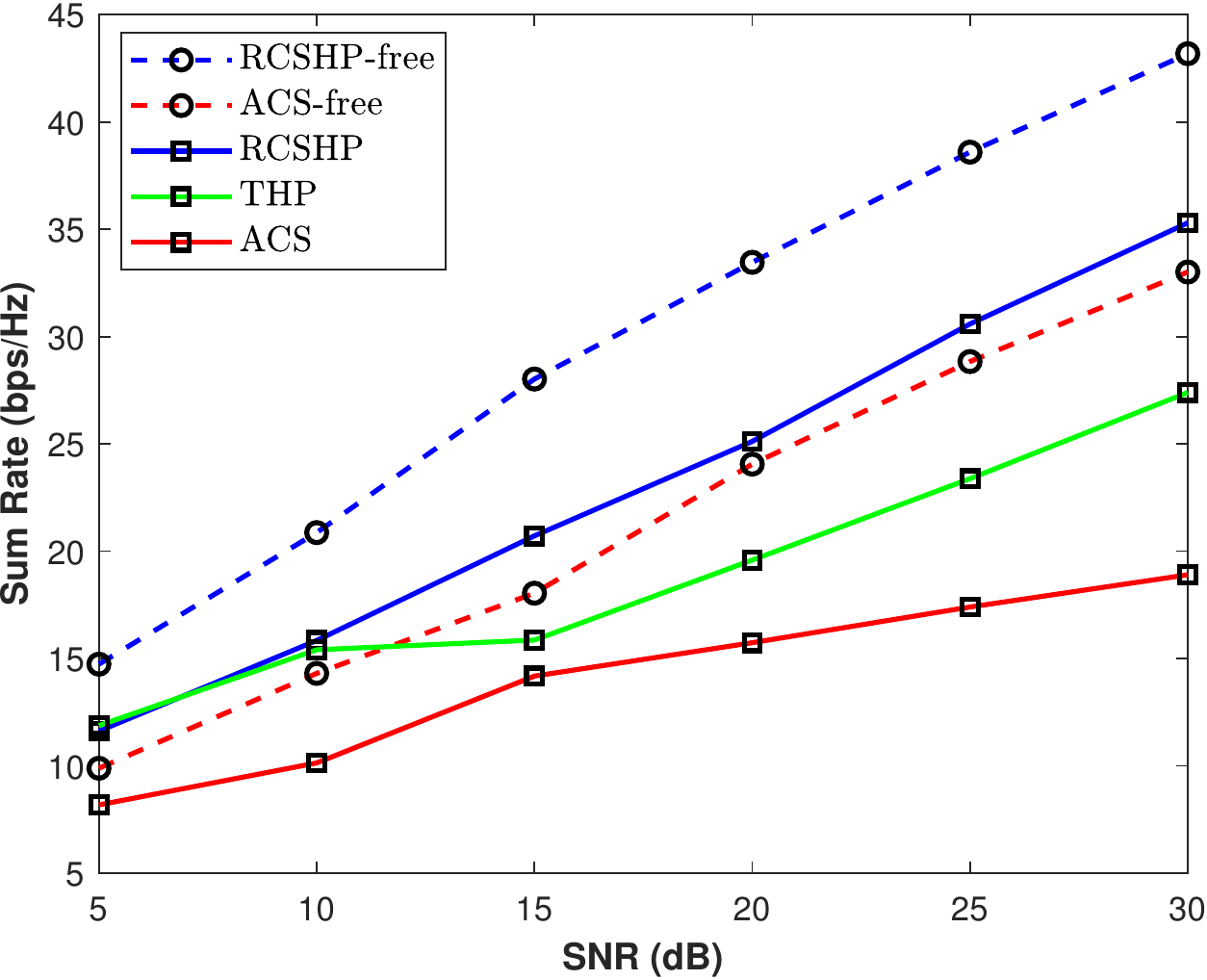}%
    		\label{SumRate_vs_Pilots_Chang_SNR} }
    	\caption{ \textbf{(a)} Sum rate versus the $ SNR $ in the COST 2100 channel. \textbf{(b)} Sum rate versus the $ SNR $ in the geometry-based channel.}
    \end{figure}

	\subsection{Proportional Fairness}
	In \figurename{~\ref{PFS_vs_Pilots_Caire}} and \figurename{~\ref{PFS_vs_Pilots_Chang}}, it can be seen that the performance gap between the RCSHP and other schemes becomes larger than that when considering the sum rate, because the randomized control policy employed in the RCSHP can realize a better fairness among users.
	\begin{figure}[!t]
		\centering
		\subfloat[]{\includegraphics[width=2.5in]{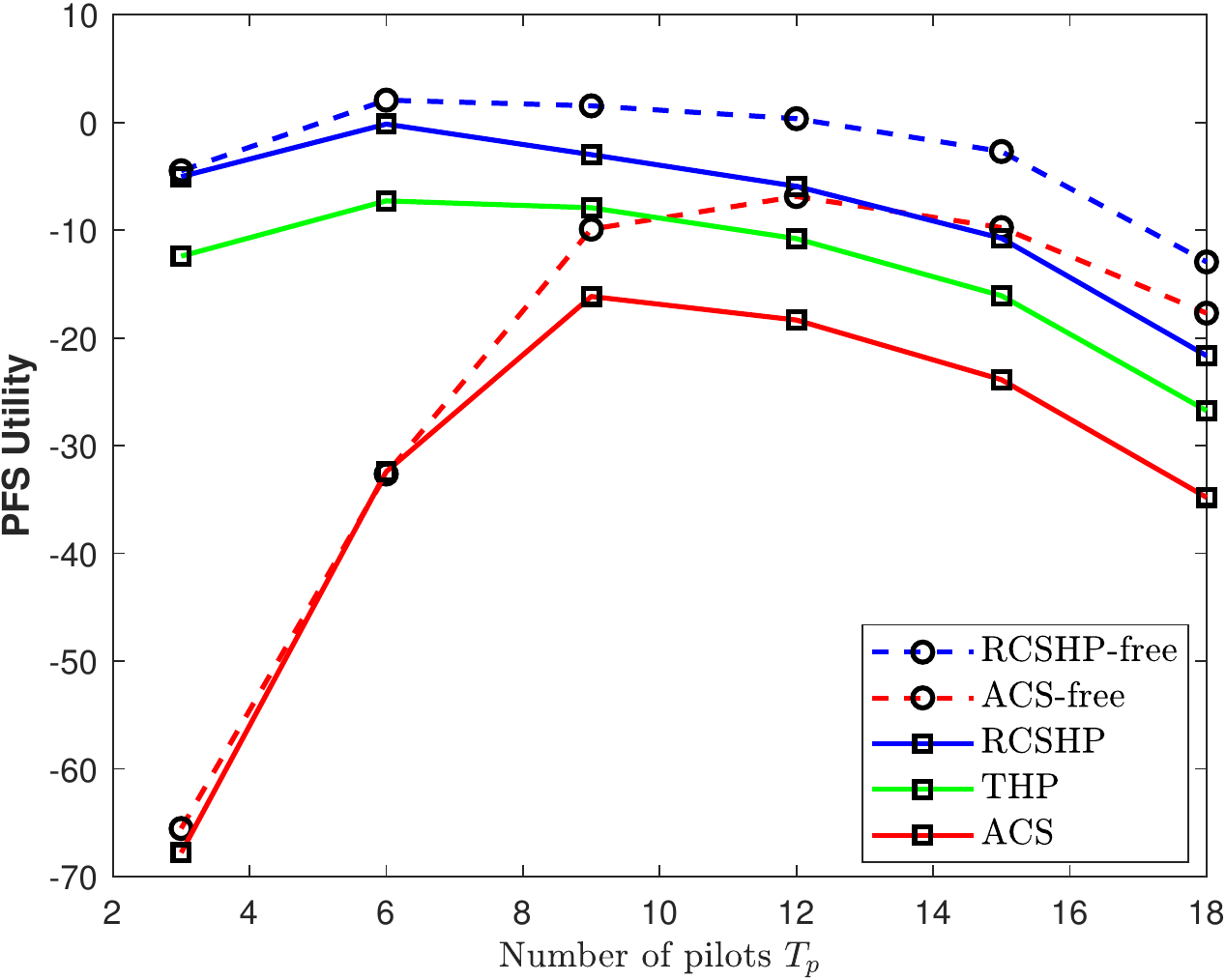}%
			\label{PFS_vs_Pilots_Caire} }
		\hfil
		\subfloat[]{\includegraphics[width=2.5in]{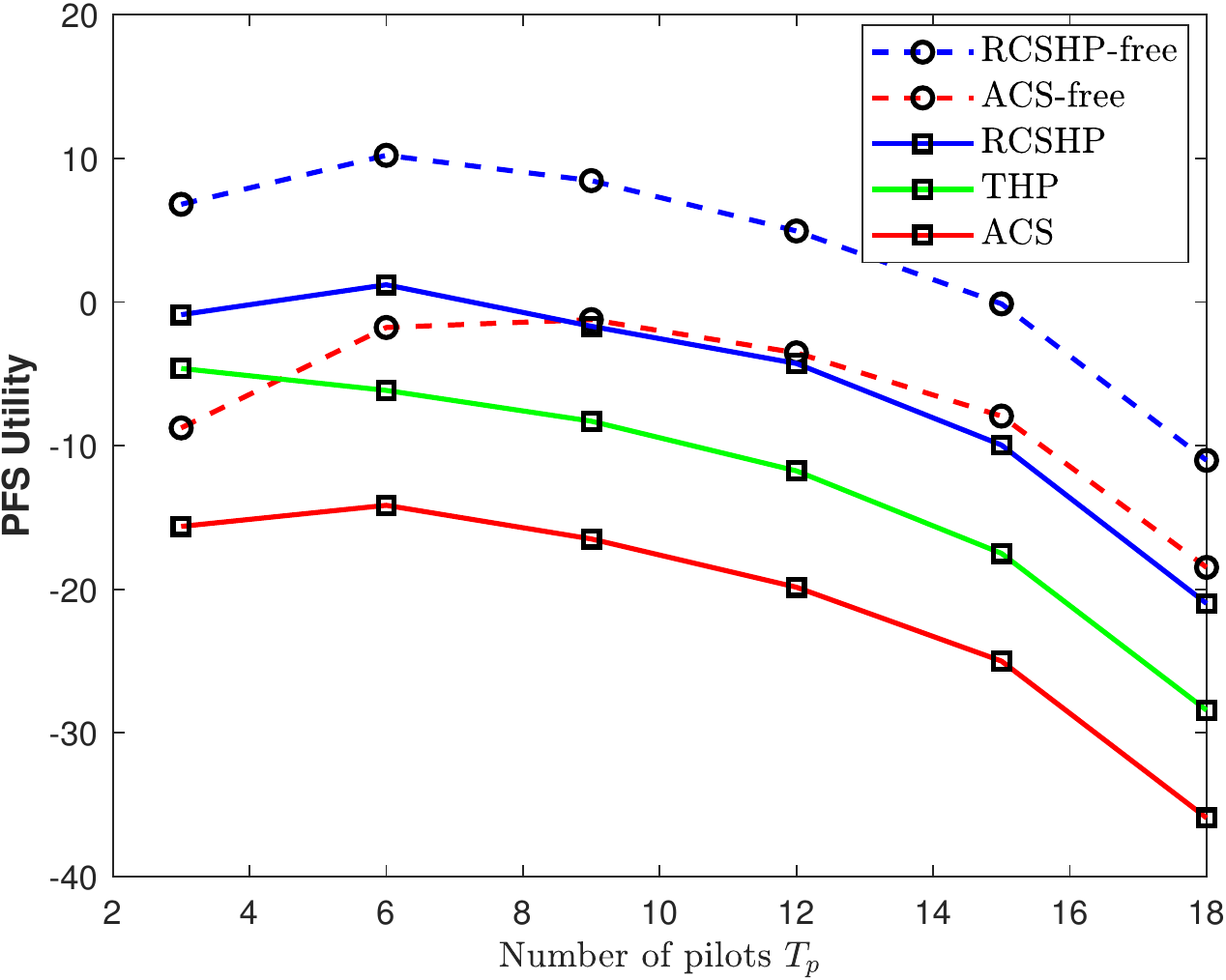}%
			\label{PFS_vs_Pilots_Chang} }
		\caption{ \textbf{(a)} PFS utility versus the number of pilots in the COST 2100 channel. \textbf{(b)} PFS utility versus the number of pilots in the geometry-based channel.}
	\end{figure}

  \subsection{Gains from CSI Errors}
  We provide the RCSHP scheme with perfect effective CSI as a baseline to show the gains from CSI errors. The sum rate and PFS utility simulated in the two channel models are given in \figurename{~\ref{Sumrate_vs_Pilots_Perfect}} and \figurename{~\ref{PFS_vs_Pilots_Perfect}}, respectively.  Note that we do not calculate the impact of pilot cost on the data rate for precoding schemes with imperfect CSI in these simulations, in order to more clearly show the effect of CSI errors. The RCSHP scheme with perfect CSI is denoted as `RCSHP-Perfect'. It can be observed from these simulations that the gains from CSI errors are significant and of value to consider. Moreover, The performance of RCSHP can significantly approach that with perfect CSI when the number of assigned pilots is enough, which indicates that the proposed scheme has a strong ability to alleviate the performance loss due to the CSI errors.
  	\begin{figure}[!t]
  	\centering
  	\subfloat[]{\includegraphics[width=2.5in]{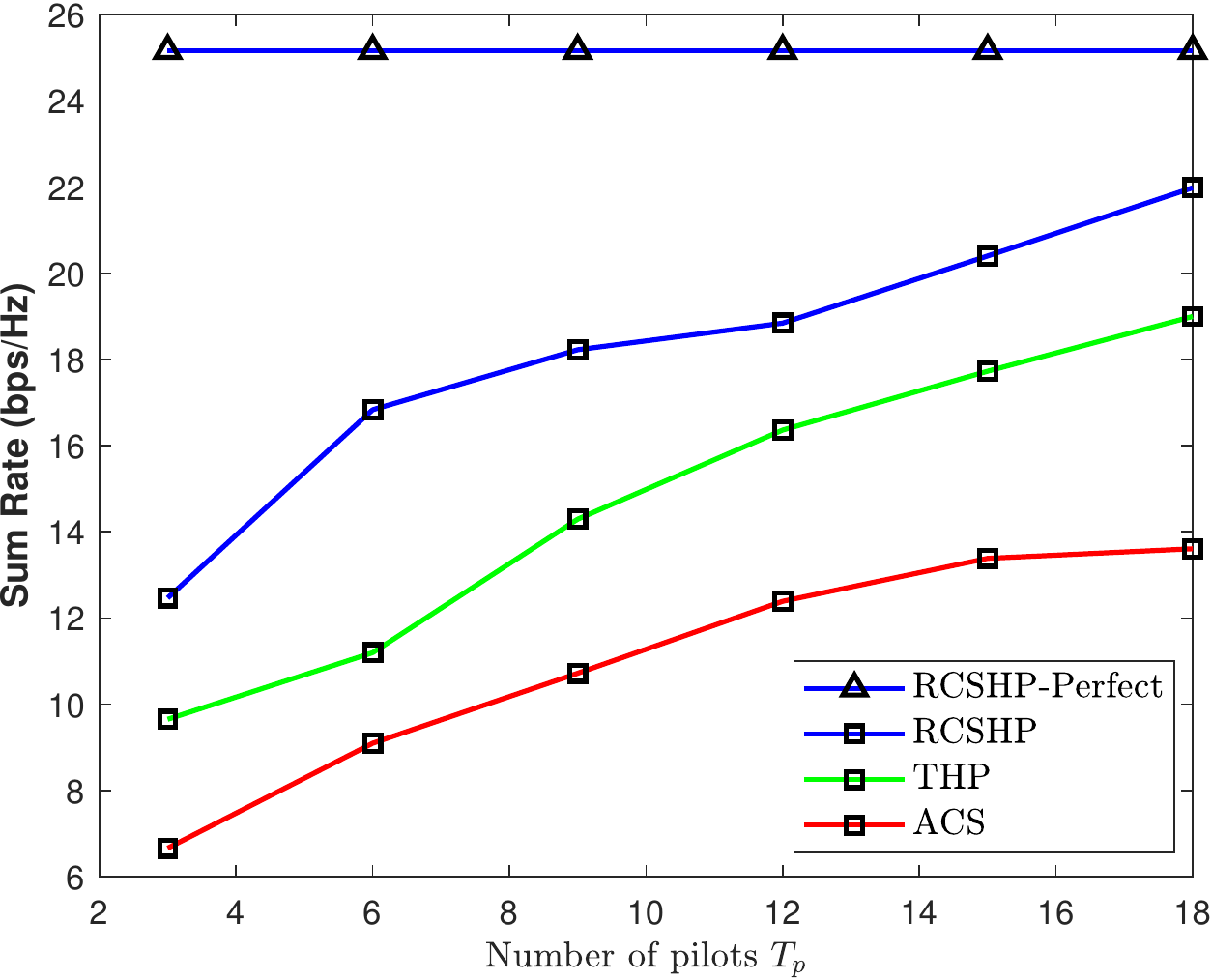}%
  		\label{Sumrate_vs_Pilots_Caire_Perfect} }
  	\hfil
  	\subfloat[]{\includegraphics[width=2.5in]{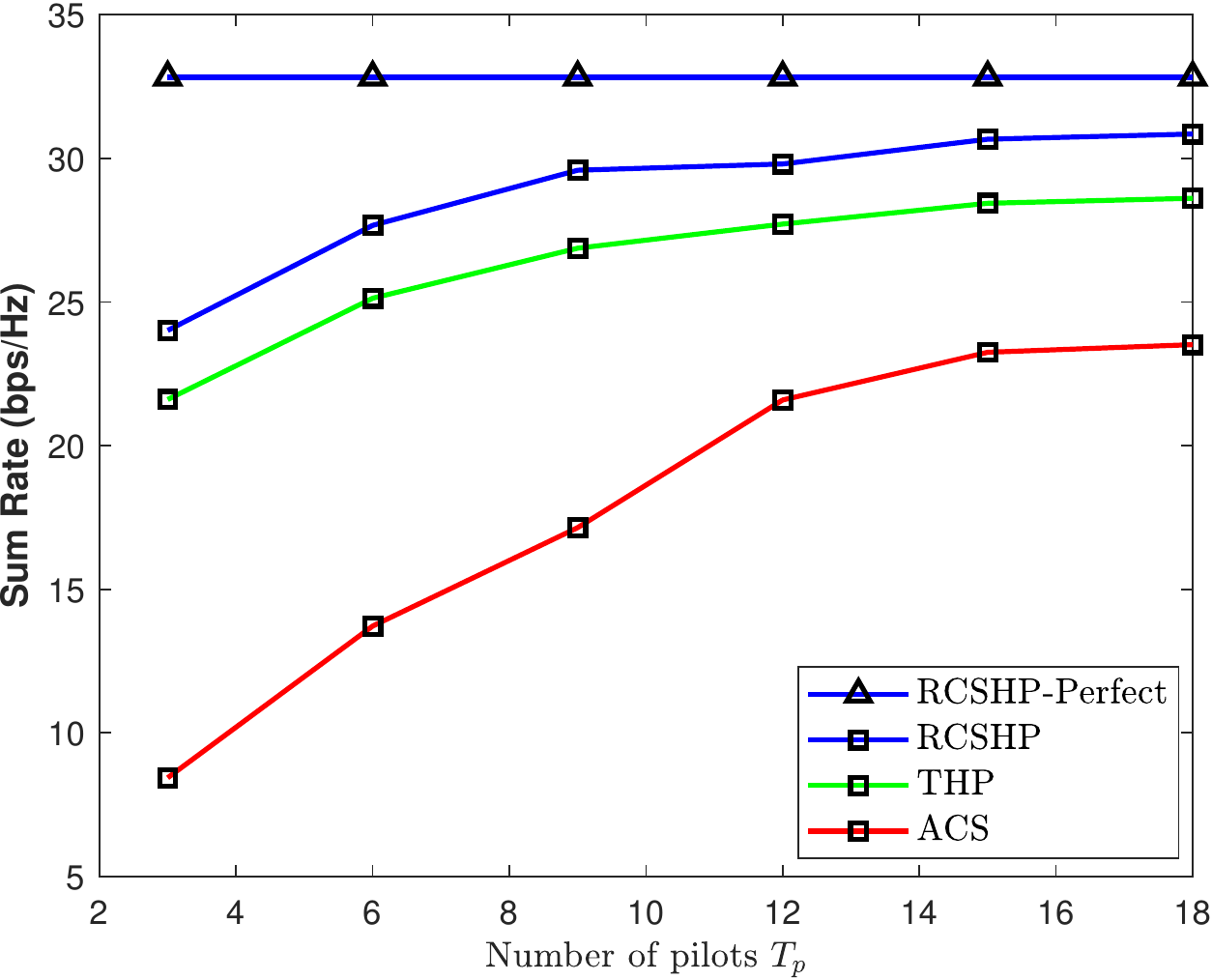}%
  		\label{Sumrate_vs_Pilots_Chang_Perfect} }
  	\caption{ \textbf{(a)} Sum rate versus the number of pilots in the COST 2100 channel (with perfect CSI as a baseline). \textbf{(b)} Sum rate versus the number of pilots in the geometry-based channel (with perfect CSI as a baseline).}
  	\label{Sumrate_vs_Pilots_Perfect}
  \end{figure}

  	\begin{figure}[!t]
	\centering
	\subfloat[]{\includegraphics[width=2.5in]{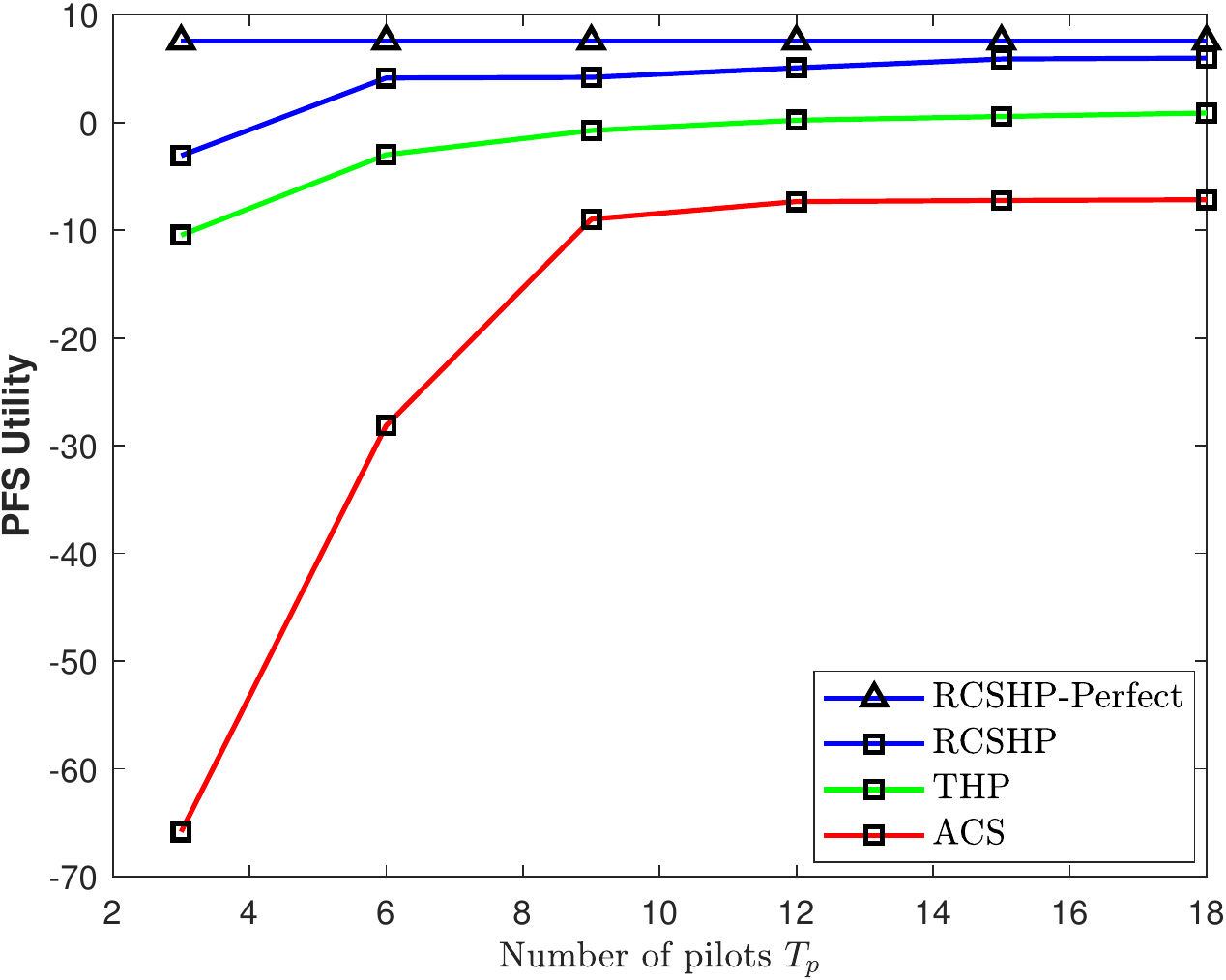}%
		\label{PFS_vs_Pilots_Caire_Perfect} }
	\hfil
	\subfloat[]{\includegraphics[width=2.5in]{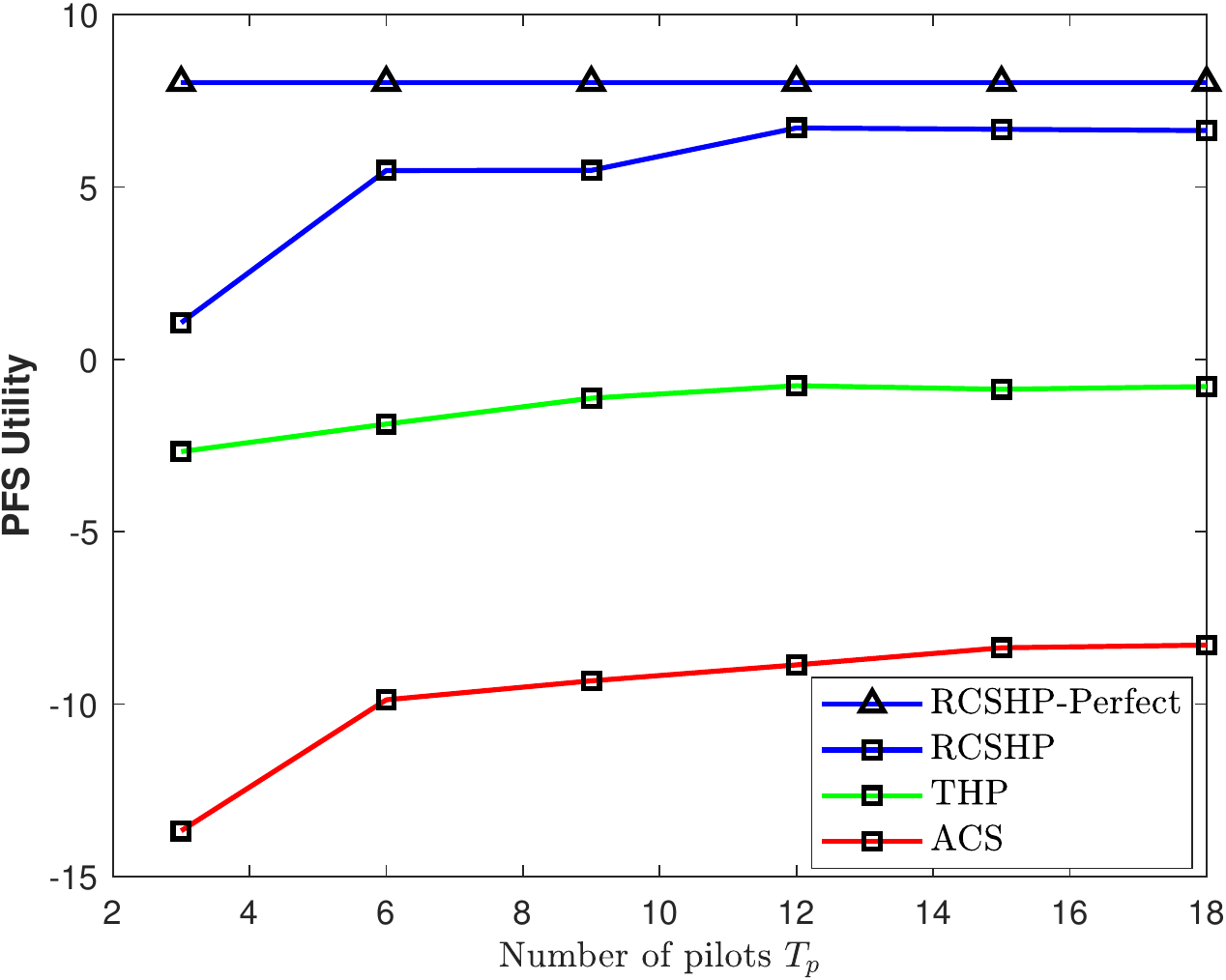}%
		\label{PFS_vs_Pilots_Chang_Perfect} }
	\caption{ \textbf{(a)} PFS utility versus the number of pilots in the COST 2100 channel (with perfect CSI as a baseline). \textbf{(b)} PFS utility versus the number of pilots in the geometry-based channel (with perfect CSI as a baseline).}
	\label{PFS_vs_Pilots_Perfect}
\end{figure}

  \subsection{Energy Efficiency}
   We provide the energy efficiency that changes with the number of pilots in the COST 2100 and geometry-based channel, respectively. As in \cite{EnergyEfficiency1_2016Heath}, we model the total power consumption at the BS as $ P_{\mathrm{tot}} = MS P_{\mathrm{PS}} + S \left( P_{\mathrm{LNA}} + P_{\mathrm{RF}} + P_{\mathrm{ADC}} \right) + P_{\mathrm{BB}} + P_{\mathrm{TX}} $, where $ P_{\mathrm{PS}} $ denotes the power consumed per phase shifter,  $ P_{\mathrm{LNA}} $ denotes the power consumed per low noise amplifier (LNA), $ P_{\mathrm{RF}} $ denotes the power consumed per RF chain,  $ P_{\mathrm{ADC}} $ is the power consumed by a single analog-to-digital converter (ADC), $ P_{\mathrm{BB}} $ denotes the power consumption of the baseband precoder and $ P_{\mathrm{TX}} $ is the consumed transmission power. Following the pioneering works on massive MIMO systems, we also model the power consumption of a baseband precoder as $ P_{\mathrm{BB}}=S\xi+\varsigma $, where $ \xi $ denotes the circuit power that scales with the number of RF chains, and $ \varsigma $ is a static circuit power term. According to \cite{EnergyEfficiency1_2016Heath,EnergyEfficiency2_2014,EnergyEfficiency3_2018}, we have $ P_{\mathrm{LNA}}=20 $ mW, $ P_{\mathrm{PS}}=6.6 $ mW, $ P_{\mathrm{RF}}=120 $ mW, $ P_{\mathrm{ADC}}=240 $ mW, $ \xi=10 $ mW/RF chain, and $ \varsigma=136 $ mW. According to {\figurename~\ref{fig:EE_Caire}} and {\figurename~\ref{fig:EE_Chang}},  we can observe that the proposed RCSHP scheme enables to achieve a better energy efficiency over other precoding schemes, which indicates that the RCSHP design can effectively reduce the power consumption.  
  
    \begin{figure}[!t]
  	\centering
  	\subfloat[]{\includegraphics[width=2.5in]{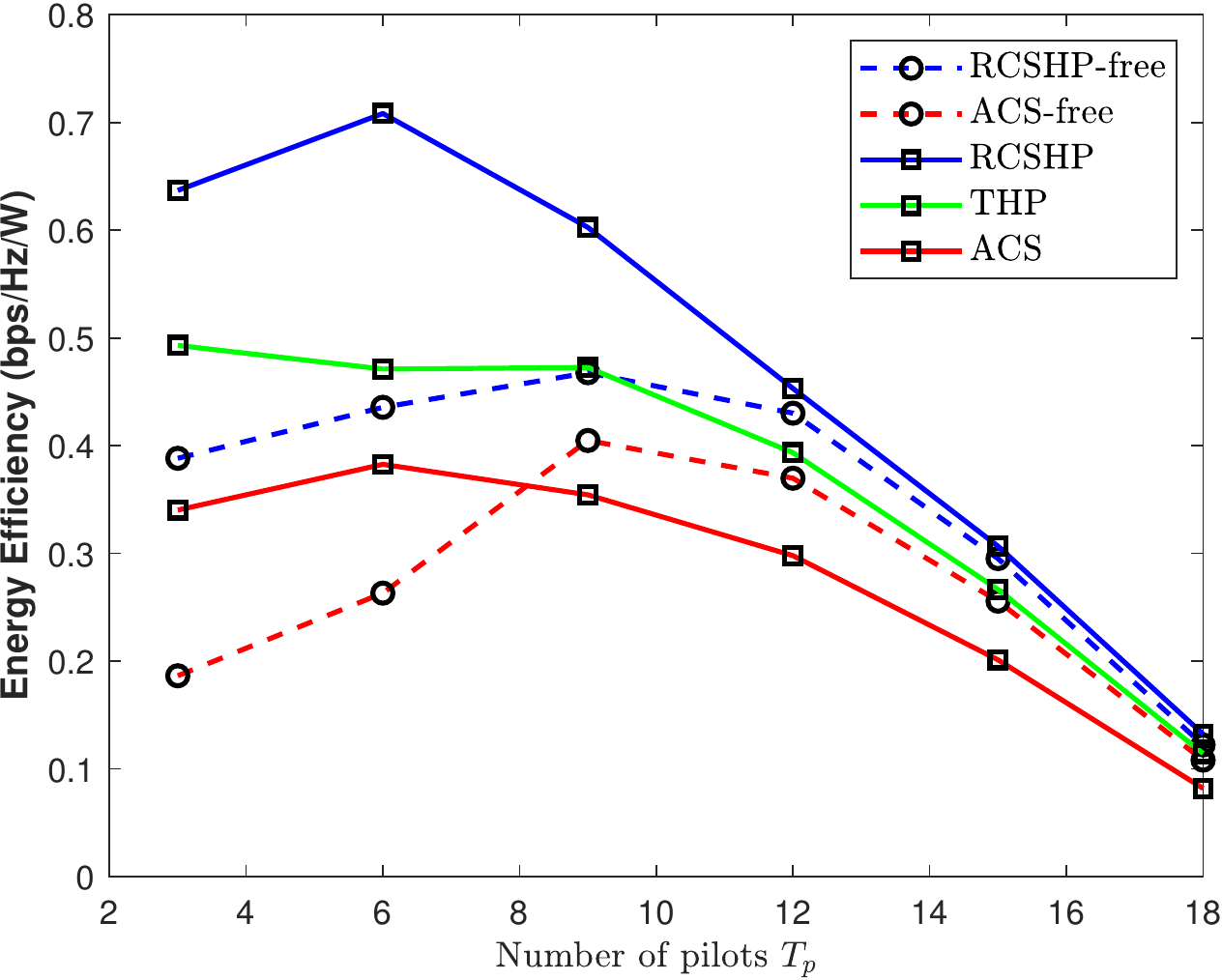}%
  		\label{fig:EE_Caire} }
  	\hfil
  	\subfloat[]{\includegraphics[width=2.5in]{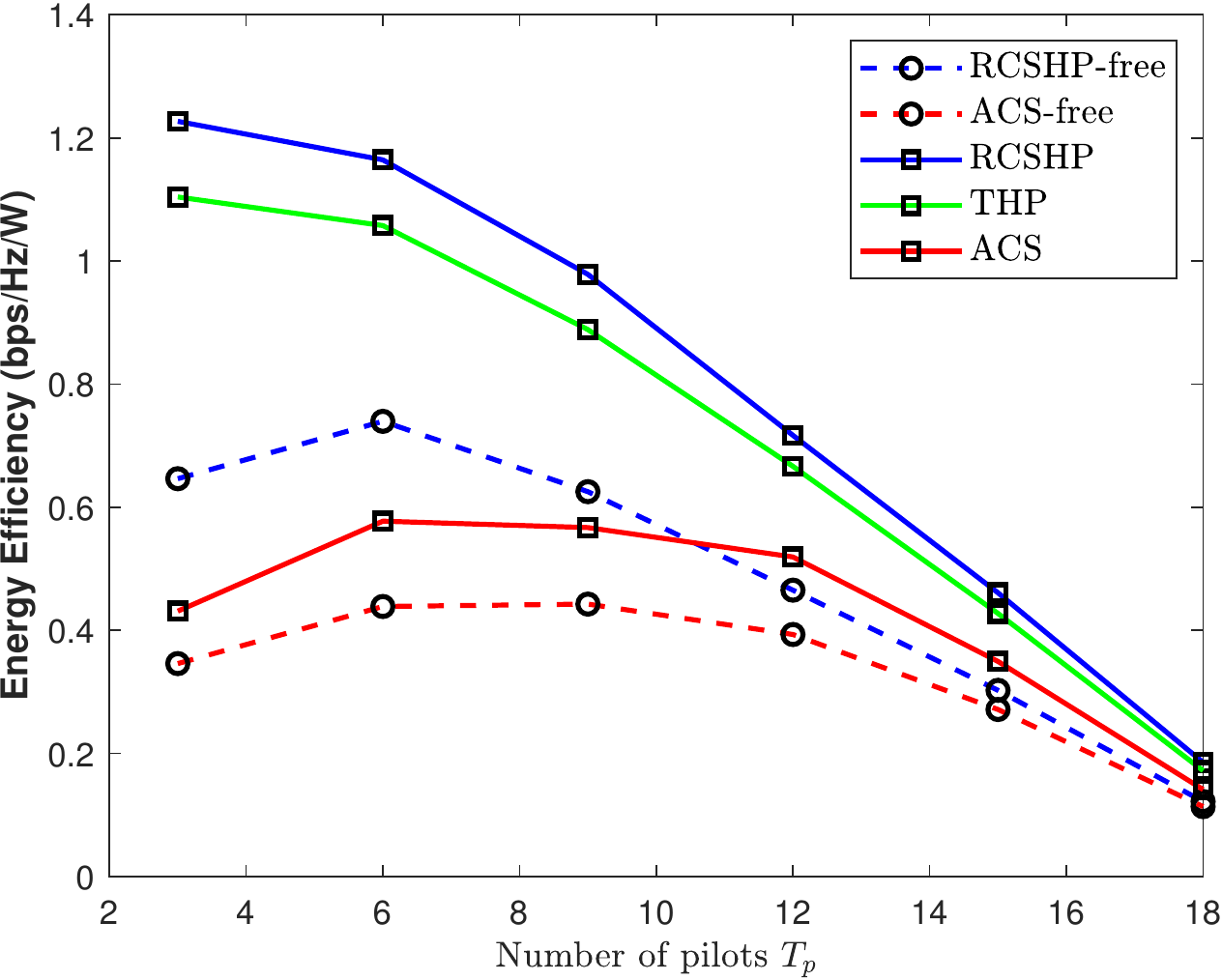}%
  		\label{fig:EE_Chang} }
  	\caption{ \textbf{(a)} Energy efficiency in the COST 2100 channel. \textbf{(b)} Energy efficiency in the geometry-based channel.}
  \end{figure}
 
	\section{Conclusion}
	We propose a novel RCSHP design for a practical FDD massive MIMO downlink transmission scenario, in which the channel environment may not be sparse and the number of assigned pilots is limited. The design is formulated as a general utility optimization problem over a randomized control policy, which is solved at the beginning of each coherence time of channel statistics, by using the proposed SSCA-RCSHP algorithm with the knowledge of channel statistics of users.  Then we apply the optimized randomized control policy to the current coherence time of channel statistics. At each time slot, the BS only needs to estimate the effective CSI of scheduled users to exploit the spatial multiplexing gain by utilizing a duality-based digital precoder. The RCSHP design can automatically group users such that in each user group, the effective CSI is sparse enough to be well estimated by using the limited number of pilots, and also has enough spatial DoF to support the simultaneous transmission to these users. Finally, extensive simulations verify that the RCSHP scheme enables to effectively realize the \textit{active channel sparsification} and achieve a significant performance gain over other baselines.  
	
	\appendices
	\section{Jacobian Matrix of Instantaneous Rate}
	First we define some useful notations: $ \mat{I}_N $ denotes a $ N\times N $ identity matrix.  $ \mat{E}_{ii}\in\R^{K\times K} $ is a matrix whose $ \roundBrack{i,i} $-th element is one and all other elements are zero. $ \mat{E}_{i}^{T_p}\in\R^{KT_p\times T_p} $ is a matrix whose $ \roundBrack{i-1}T_p + 1 $-th row to $ iT_p $-th row is stacked into an identity matrix $ \mat{I}_{T_p} $. Similarly, $ \mat{E}_{ii}^{T_p}\in\R^{KT_p\times KT_p} $ is a matrix whose $ \roundBrack{i-1}T_p + 1 $-th row to $ iT_p $-th row  and $ \roundBrack{i-1}T_p + 1 $-th column to $ iT_p $-th column is put into an identity matrix $ \mat{I}_{T_p} $. $ \mat{E}_{i}^{M}\in\R^{KM\times M} $ is a matrix whose $ \roundBrack{i-1}M + 1 $-th row to $ iM $-th row is put into an identity matrix $ \mat{I}_{M} $. Define an aggregated covariance matrix $ \mat{C} = \squareBrack{\mat{C}_1, \dots, \mat{C}_K}\in\C^{M\times KM} $ and an aggregated channel estimation noise matrix $ \mat{N} = \squareBrack{ \vect{n}_1,\dots,\vect{n}_K }^{\T} \in \C^{K\times T_p} $, we have $ \mat{A} = \mat{F}^\CT\mat{C} \left(\mat{I}_K\otimes\mat{F} \mat{\Psi }^\T\right), $
	$ \mat{B} = \sum_{i=1}^K \mat{E}_{i}^{T_p}  \left( \mat{\Psi }^{*}  \mat{A} + \vect{1}_K^{\T} \otimes \mat{I}_{T_p} \right) \mat{E}_{ii}^{T_p}, $
	$ \mat{Y} = \sum_{i=1}^K \mat{E}_{i}^{T_p} \left( \mat{HF} \mat{\Psi }^{\T} + \mat{N} \right)^{\CT} \mat{E}_{ii}, $ $ \hat{ \mat{G} } = \sum_{i=1}^K \mat{E}_{i}^{M} \bar{ \mat{G} } \mat{E}_{ii} $, where $ \vect{1}_K\in\R^{K\times 1} $ is an all-one vector. Furthermore, denote $ \roundBrack{ \mat{Z} }_{\Blk} $ as a block matrix whose all $ M\times S $ blocks are matrix $ \mat{Z} $, then define $ \roundBrack{ \mat{Z}_{nm} }_{\Blk} $ to represent a block matrix whose $ \roundBrack{m,n} $-th block is matrix $ \mat{Z}_{nm} $, where $ m=1,\dots,M $ and $ n=1,\dots,S $. Denote $ \mat{J}_{mn}\in\R^{ M\times S } $ is a matrix whose $ \roundBrack{m,n} $-th element is one and all other elements are zero, and $ \bar{ \mat{J} }_{nm} = \mat{J}_{mn}^\T $. Then we have $ \bar{ \mat{E} }_h =\roundBrack{ \bar{ \mat{J} }_{nm} }_{\Blk} \in \R^{ MS\times MS } $, $ \bar{ \mat{F} }_h =\roundBrack{ \mat{F}^\CT }_{\Blk}\in \C^{ MS\times MS } $, $ \widetilde{ \mat{F} } = \roundBrack{ \mat{I}_K \otimes \sqrt{-1}\squareBrack{ \mat{F} }_{mn} \mat{J}_{mn} }_{\Blk} \in \C^{ KM^2\times KS^2 } $. Denote $ \BlkTran\roundBrack{ \mat{A}_1, m,n } $ as a operator  that evenly divide a matrix $ \mat{A}_1 \in \C^{M_1\times N_1} $ into $ mn $ blocks, then treat each block as an element and transpose the divided matrix into a new one. The size of each block is $ \frac{M_1}{m} \times \frac{N_1}{n} $. Then we introduce some notations:
		\begin{equation*}
		\begin{split}
		\bar{ \mat{A} } &= -\sqrt{-1}\roundBrack{ \bar{ \mat{F} }_h \circ \bar{ \mat{E} }_h } \roundBrack{ \mat{I}_S \otimes \mat{C}\roundBrack{ \mat{I}_K \otimes \mat{F} \mat{\Psi }^{\T}  } } + \roundBrack{ \mat{I}_M \otimes \mat{F}^\CT \mat{C} } \widetilde{ \mat{F} } \roundBrack{ \mat{I}_{SK} \otimes  \mat{\Psi }^{\T}  },\\
		\bar{ \mat{H} }_h &= \bar{ \mat{A} } \roundBrack{ \mat{I}_S \otimes \mat{B}^{-1}\mat{Y} } + \sum_{i=1}^K \roundBrack{ \mat{I}_M \otimes \mat{A}\mat{B}^{-1} \mat{E}_{i}^{T_p} \mat{\Psi }^{*}  } \\
		& \roundBrack{ -\sqrt{-1}\roundBrack{ \bar{ \mat{F} }_h \circ \bar{ \mat{E} }_h } \roundBrack{ \mat{I}_S \otimes \mat{H}^\CT \mat{E}_{ii} } - \bar{ \mat{A} } \roundBrack{ \mat{I}_S \otimes \mat{E}_{ii}^{T_p} \mat{B}^{-1}\mat{Y} } }.
		\end{split}
		\end{equation*}
	
	
	Further, define $ \bar{ \mat{H} } = \BlkTran\roundBrack{ \bar{ \mat{H} }_h^{\CT},S,M } $, $ \mat{V} = \hat{ \widetilde{ \mat{H} } }^\CT \mat{P} \hat{ \widetilde{ \mat{H} } } + \mat{I} $ and $ \widetilde{ \mat{G} } =  \mat{V}^{-1} \hat{ \widetilde{ \mat{H} } }^\CT \mat{P} $, we have
	\begin{equation*}
	\begin{split}
	\bar{ \mat{D} }_h &= \roundBrack{ \mat{I}_M \otimes \roundBrack{ \mat{P} - \mat{P} \hat{ \widetilde{ \mat{H} } } \widetilde{ \mat{G} } } } \bar{ \mat{H} } \roundBrack{ \mat{I}_S \otimes \mat{V}^{-1} \mat{F}^\CT }  - \roundBrack{ \mat{I}_M \otimes \widetilde{ \mat{G} }^\CT } \roundBrack{ \sqrt{-1}\roundBrack{ \bar{ \mat{F} }_h \circ \bar{ \mat{E} }_h } + \bar{ \mat{H} }_h \roundBrack{ \mat{I}_S \otimes \bar{ \mat{G} }^\CT }  },\\
	\mat{A}_F &= \sum_{i=1}^K -\roundBrack{ \mat{I}_M \otimes \mat{ \Lambda } \mat{E}_{ii} } \bar{ \mat{D} }_h \roundBrack{ \mat{I}_S \otimes \frac{1}{2} \roundBrack{ \mat{\Lambda}^{\frac{1}{2}} \hat{ \mat{G} }^\CT \mat{E}_i^M }^\CT } - \roundBrack{ \mat{I}_M \otimes \mat{ \Lambda } \hat{ \mat{G} }^\CT \mat{E}_{i}^M } \bar{ \mat{D} } \roundBrack{ \mat{I}_S \otimes \frac{1}{2} \mat{E}_{ii} \mat{\Lambda}^{\frac{1}{2}} },\\
	\mat{G}_F &= \roundBrack{ \mat{I}_M \otimes \mat{V}^{-1} } \bar{ \mat{H} }_h \roundBrack{ \mat{I}_S \otimes \roundBrack{ \mat{P} \mat{\Lambda}^{\frac{1}{2}} - \mat{P} \hat{ \widetilde{ \mat{H} } } \mat{G} } } + \roundBrack{ \mat{I}_M \otimes \widetilde{ \mat{G} } } \roundBrack{ \mat{A}_F - \bar{ \mat{H} } \roundBrack{ \mat{I}_S \otimes \mat{G} } },
	\end{split}
	\end{equation*}	
	where $ \bar{ \mat{D} } = \BlkTran\roundBrack{ \bar{ \mat{D} }_h^{\CT}, S,M } $. According to the matrix calculus, the gradient of instantaneous data rate $ r_k\roundBrack{\pmb{\theta}\roundBrack{l},\vect{p}\roundBrack{l}; \mat{H},\mat{N},l} $ w.r.t. $ \pmb{\theta}\roundBrack{l} $ is given by
	\begin{equation}
	\nabla_{ \pmb{\theta}\roundBrack{l} } r_k = \frac{ \sum_{i=1}^K \vect{a}^{\theta_l}_{k,i} }{\Gamma_k}-\frac{ \sum_{i\neq k}\vect{a}^{\theta_l}_{k,i} }{\Gamma_{-k}},
	\end{equation}
	where $ \Gamma_k=\sum_{i=1}^K p_i\abs{\vect{h}_k^\CT \mat{F}\vect{g}_i}^2 + 1 $, $ \Gamma_{-k}=\sum_{i\neq k} p_i\abs{\vect{h}_k^\CT \mat{F}\vect{g}_i}^2 + 1 $, 
	\begin{equation*}
	\vect{a}^{\theta_l}_{k,i} = \Vectorize \Bigg \{ 2p_i \Re \bigg [ \squareBrack{ \mat{HFG} }_{ki}^{*} \Big ( \roundBrack{ \mat{I}_M \otimes \squareBrack{\mat{H}}_{k.} } \roundBrack{ \sqrt{-1}\bar{ \mat{F} } \circ \bar{ \mat{E} } } 
    \roundBrack{ \mat{I}_S \otimes \squareBrack{\mat{G}}_{.i} }  + \roundBrack{ \mat{I}_M \otimes \squareBrack{\mat{HF}}_{k.} } \mat{G}_F \roundBrack{ \mat{I}_S \otimes \vect{e}_i } \Big )  \bigg ] \Bigg \} ,
    \end{equation*}
	where $ \vect{e}_i \in \R^{K\times 1} $ is a vector whose $ i $-th element is one and all other elements are zero.
	
	Note that we have omitted the control state $ l $ and $ \roundBrack{\pmb{\theta}\roundBrack{l},\vect{p}\roundBrack{l}; \mat{H},\mat{N},l} $ in the gradient expression  for simplicity, and we also keep this habit for the gradient of $ r_k\roundBrack{\pmb{\theta}\roundBrack{l},\vect{p}\roundBrack{l}; \mat{H},\mat{N},l}  $ over $ \vect{p}\roundBrack{l} $. $ \bar{ \mat{P} } \in \R^{K^2 \times K} $ is a matrix whose $ \roundBrack{i-1}K+1 $-th row to $ iK $-th row is put into $ \mat{E}_{ii} $, where $ i=1,\dots,K $. Then we introduce some notations:
	\begin{equation*}
	\begin{split}
	\mat{D}_p &= \roundBrack{ \mat{I}_K \otimes \mat{F} \mat{V}^{-1} \hat{ \widetilde{ \mat{H} } }^\CT } \roundBrack{ \bar{ \mat{P} } - \bar{ \mat{P} } \hat{ \widetilde{ \mat{H} } } \widetilde{ \mat{G} } },\\
	\mat{A}_p &= \sum_{i=1}^K -\frac{1}{2} \roundBrack{ \mat{I}_K \otimes \mat{ \Lambda } \mat{E}_{ii} } \bar{ \mat{D} }_p \roundBrack{ \hat{ \mat{G} }^\CT \mat{E}_i^M }^\CT \mat{\Lambda}^{\frac{1}{2}} -\frac{1}{2} \roundBrack{ \mat{I}_K \otimes \mat{\Lambda} \hat{ \mat{G} }^\CT \mat{E}_i^M } \mat{D}_p \mat{E}_{ii} \mat{\Lambda}^{\frac{1}{2}},\\
	\mat{G}_p &= \roundBrack{ \mat{I}_K \otimes \mat{V}^{-1} \hat{ \widetilde{ \mat{H} } }^\CT } \bar{ \mat{P} } \roundBrack{ \mat{\Lambda}^{\frac{1}{2}} - \hat{ \widetilde{ \mat{H} } } \mat{G} } + \roundBrack{ \mat{I}_K \otimes \widetilde{ \mat{G} } } \mat{A}_p,
	\end{split}
	\end{equation*}
	where $ \bar{ \mat{D} }_p = \BlkTran\roundBrack{ \mat{D}_p^{\CT},1,K } $. Thus using matrix calculus, the gradient of instantaneous data rate $ r_k\roundBrack{\pmb{\theta}\roundBrack{l},\vect{p}\roundBrack{l}; \mat{H},\mat{N},l}  $ w.r.t $ \vect{p}\roundBrack{l} $ is given by
	\begin{equation}
	\nabla_{ \vect{p}\roundBrack{l} } r_k =  \frac{ \sum_{i=1}^K \vect{a}^{p_l}_{k,i} }{\Gamma_k}-\frac{ \sum_{i\neq k} \vect{a}^{p_l}_{k,i} }{\Gamma_{-k}},
	\end{equation}
	where 
	\begin{equation*}
	\vect{a}^{p_l}_{k,i} = 2p_i \Re \Big [ \squareBrack{ \mat{HFG} }_{ki}^{*} \roundBrack{ \mat{I}_K \otimes \squareBrack{\mat{HF}}_{k.} } \squareBrack{ \mat{G}_p }_{.i}  \Big ]
	+ \vect{e}_i \abs{ \squareBrack{ \mat{HFG} }_{ki} }^2.
	\end{equation*}
	
	Therefore, for given channel state $ \mat{H} $ and channel estimation noise state $ \mat{N} $, the Jacobian matrix of the data rate vector $ \tilde{\vect{r}}\roundBrack{\vect{\Gamma},\vect{q}; \mat{H},\mat{N}} $ w.r.t. $ \vect{\Gamma} $ is
	\begin{equation}
	\mathbf{J}_{\vect{\Gamma}}\roundBrack{\vect{\Gamma},\vect{q}; \mat{H},\mat{N}}
	=\begin{bmatrix}
	q_1\nabla_{\pmb{\theta}\roundBrack{1}}r_1 & \cdots & q_1\nabla_{\pmb{\theta}\roundBrack{1}}r_K\\
	q_1\nabla_{\vect{p}\roundBrack{1}}r_1     & \cdots & q_1\nabla_{\vect{p}\roundBrack{1}}r_K\\
	\vdots                                    &        & \vdots  \\
	q_L\nabla_{\pmb{\theta}\roundBrack{L}}r_1 & \cdots & q_L\nabla_{\pmb{\theta}\roundBrack{L}}r_K\\
	q_L\nabla_{\vect{p}\roundBrack{L}}r_1     & \cdots & q_L\nabla_{\vect{p}\roundBrack{L}}r_K\\
	\end{bmatrix}.
	\end{equation}

	\section{Proof of Lemma 1}
	 The proof relies on the following lemma.
		\begin{lemma}
			\textit{	Define $ \breve{\bar{r}}^t_k\roundBrack{l}=\Expect\nolimits_{\mat{H},\vect{N}}\squareBrack{r_k\roundBrack{\vect{\Gamma}^t\roundBrack{l};\mat{H},\vect{N},l}} $, then under Assumption 1, we have}
			\begin{align}
			\lim_{t\rightarrow\infty}\abs{\hat{r}^t_{k}\roundBrack{l}-\breve{\bar{r}}^t_k\roundBrack{l}}&=0, \forall k, \forall l,\label{FuncConverged} \\ 
			\lim_{t\rightarrow\infty} \norm{ \mathbf{f}^t_{ \vect{\Gamma} } - \nabla_{ \vect{\Gamma} } f\roundBrack{ \vect{\Gamma}^t, \vect{q}^t} }&=0. \label{GradCoverged}
			\end{align} 
		\end{lemma}
		
		\begin{proof}
			For \eqref{FuncConverged}, it is a consequence of \cite{1980FeasibleDirection}, Lemma 1. It is easy to verify that the technical conditions (a), (b), (c) and (d) therein are satisfied. Moreover, it follows from the Lipschitz continuity of $ r_k\roundBrack{\vect{\Gamma}^t\roundBrack{l};\mat{H},\vect{N},l} $ that
			\begin{equation*}
			\lim_{t\rightarrow\infty}\frac{\abs{\breve{\bar{r}}^{t+1}_k\roundBrack{l}-\breve{\bar{r}}^t_k\roundBrack{l}}}{\rho_t} \leq\lim_{t\rightarrow\infty} \frac{L_{\Gamma}\norm{\vect{\Gamma}^{t+1}\roundBrack{l}-\vect{\Gamma}^t\roundBrack{l}}}{\rho_t}
			=\lim_{t\rightarrow\infty}O\roundBrack{\frac{\gamma_t}{\rho_t}}=0,
			\end{equation*}
			where $ L_{\Gamma} $ is the Lipschitz constant. Therefore, the technical condition (e) in \cite{1980FeasibleDirection} is also satisfied and \eqref{FuncConverged} is proved. 
			
			The proof of \eqref{GradCoverged} consists of two steps. For the consideration of simplicity, let $ \vect{\bar{r}}^t=\vect{\bar{r}}\roundBrack{\vect{\Gamma}^t, \vect{q}^t} $, $ \vect{\hat{\bar{r}}}^t=\vect{\hat{\bar{r}}}^t\roundBrack{\vect{q}^t} $ and $ \nabla^t_{\vect{\Gamma}}f= \nabla_{\vect{\Gamma}}f\roundBrack{\vect{\Gamma}^t, \vect{q}^t} $. 
			
			\textit{Step 1 of proving \eqref{GradCoverged}:} Define a sequence
			\begin{equation}
			\check{r}_k^t = \sum_{l=1}^{L} q_l^t \sum_{i=1}^{t} r_k\roundBrack{\vect{\Gamma}^t\roundBrack{l}; \mat{H}\roundBrack{i},\mat{N}\roundBrack{i},l}, \forall k.
			\end{equation}
			Denote $ \vect{ \check{r} }^t = \squareBrack{ \check{r}_1^t, \dots, \check{r}_K^t }^{\T} $. According to the law of large numbers, the central limit theorem and the Berry-Esseen theorem, we have  $ \lim_{t\rightarrow\infty} \norm{ \vect{ \check{r} }^t - \vect{\bar{r}}^t }=0 $ and $ \Expect\norm{ \vect{ \check{r} }^t - \vect{\bar{r}}^t } = O\roundBrack{ \frac{1}{\sqrt{t}} } $. Further we define a sequence
			\begin{equation}\label{GradientUpdateSuro}
			\hat{ \mathbf{f} }^t_{ \vect{\Gamma} } = \roundBrack{1-\rho_t} \hat{ \mathbf{f} }^{t-1}_{ \vect{\Gamma} }
			+ \rho_t \sum_{i=1}^{T_{HN}} \dfrac{ \mathbf{J}_{ \vect{\Gamma} } \roundBrack{ \vect{\Gamma}^t,\vect{q}^t; \mat{H}^t\roundBrack{i},\mat{N}^t\roundBrack{i} } \nabla_{ \vect{\bar{r}} }U\roundBrack{ \vect{ \check{r} }^t } }{T_{HN}}.
			\end{equation}
			It can be seen that the update term in \eqref{GradientUpdateSuro}, which is denoted by $ \tilde{ \mathbf{f} }^t_{ \vect{\Gamma} } $,  is only different from \eqref{GradientUpdate} by $ e_t = \abs{ \sum_{i=1}^{T_{HN}} \dfrac{ \mathbf{J}_{ \vect{\Gamma} } \roundBrack{ \vect{\Gamma}^t,\vect{q}^t; \mat{H}^t\roundBrack{i},\mat{N}^t\roundBrack{i} } \roundBrack{ \nabla_{ \vect{\bar{r}} }U\roundBrack{ \vect{ \check{r} }^t } - \nabla_{ \vect{\bar{r}} }U\roundBrack{ \vect{\hat{\bar{r}}}^t }  } }{T_{HN}}  } $.
			
			Further we have 
			\begin{equation}
			\lim_{t\rightarrow\infty} \norm{ \hat{ \mathbf{f} }^t_{ \vect{\Gamma} } - \nabla^t_{\vect{\Gamma}}f } = 0. \label{Intermediate}
			\end{equation} 
			This is a consequence of \cite{1980FeasibleDirection}, Lemma 1. It is easy to verify that the technical conditions (a), (b), (d) and (e) in \cite{1980FeasibleDirection}, Lemma 1 are satisfied. Moreover, follow from that $ \nabla_{\vect{\bar{r}}}U $ is Lipschitz continuous and $ \mathbf{J}^t_{\vect{\Gamma}} $ is bounded w.p.1., we have 
			\begin{equation}
			\norm{ \Expect\squareBrack{ \tilde{ \mathbf{f} }^t_{ \vect{\Gamma} } }- \nabla^t_{ \vect{\Gamma} }f } \leq \Expect\norm{ \mathbf{J}^t_{ \vect{\Gamma} } \roundBrack{ \nabla_{ \vect{\bar{r}} }U\roundBrack{ \vect{ \check{r} }^t } - \nabla_{ \vect{\bar{r}} }U\roundBrack{ \vect{\bar{r}}^t } } }\nonumber
			= O\roundBrack{\norm{ \vect{ \check{r} }^t -\vect{ \bar{r} }^t } } = O\roundBrack{ \frac{1}{ \sqrt{t} } }.
			\end{equation} 
			From Assumption 1-1), we have $ \sum_{t} \rho_t \norm{ \Expect\squareBrack{ \tilde{ \mathbf{f} }^t_{ \vect{\Gamma} } }- \nabla^t_{ \vect{\Gamma} }f }<\infty $, which implies the technical condition (c) in \cite{1980FeasibleDirection}, Lemma 1 is satisfied.
			
			\textit{Step 2 of proving \eqref{GradCoverged}:} From the definitions of $ \hat{ \mathbf{f} }^t_{ \vect{\Gamma} } $ and $ \mathbf{f}^t_{ \vect{\Gamma} } $, we have
			\begin{flalign}
			&\norm{ \hat{ \mathbf{f} }^t_{ \vect{\Gamma} } - \mathbf{f}^t_{ \vect{\Gamma} } } \nonumber \leq \sum_{t^{'}=1}^t \roundBrack{ 1-\rho_t }^{ t-t^{'} } \rho_{t^{'}} e_{t^{'}} \nonumber 
			= \sum_{t^{'}=1}^{ n_{t} } \roundBrack{ 1-\rho_t }^{ t-t^{'} } \rho_{t^{'}} e_{t^{'}} + \sum_{t^{'}=n_{t}+1}^{t} \roundBrack{ 1-\rho_t }^{ t-t^{'} } \rho_{t^{'}} e_{t^{'}} \nonumber \\
			&\qquad \qquad\leq \rho_1 e_{t_a} \frac{ \roundBrack{1-\rho_t}^{t-n_t} }{\rho_t} + \frac{ \rho_{n_t + 1} }{\rho_t} e_{t_b},
			\end{flalign}
			where $ n_t = \roundBrack{1-\beta-\epsilon}t $ with $ \epsilon\in\roundBrack{0,1-\beta} $, $ e_{t_a} = \max_{ t^{'}\in\curlyBrack{1,\dots,n_t} } e_{ t^{'} } $ and $ e_{t_b} = \max_{ t^{'}\in\curlyBrack{n_t+1,\dots,t} }e_{ t^{'} } $. From Assumption 1-1), we have $ \lim_{t\rightarrow\infty} \rho_1 e_{t_a} \frac{ \roundBrack{1-\rho_t}^{t-n_t} }{\rho_t} = 0 $ and $ \frac{ \rho_{n_t + 1} }{\rho_t}<\infty $. Then it follows from the above analysis that 
			\begin{equation}
			\lim_{t\rightarrow\infty} \norm{ \hat{ \mathbf{f} }^t_{ \vect{\Gamma} } - \mathbf{f} ^t_{ \vect{\Gamma} } } = \lim_{t_b\rightarrow\infty} O\roundBrack{e_{t_b}} \overset{a}{=} \lim_{t_b\rightarrow\infty} O\roundBrack{ \norm{ \vect{ \check{r} }^{t_b} - \vect{\hat{\bar{r}}}^{t_b} } } \overset{b}{=} 0, \label{Final}
			\end{equation}
			where \eqref{Final}-$ a $ holds because $ \nabla_{\vect{\bar{r}}}U $ is Lipschitz continuous and $ \mathbf{J}^t_{\vect{\Gamma}} $ is bounded w.p.1., and \eqref{Final}-$ b $ holds because $ \lim_{t\rightarrow\infty} = \norm{ \vect{\hat{\bar{r}}}^t - \vect{\bar{r}}^t } = 0 $ and $ \lim_{t\rightarrow\infty} \norm{ \vect{\check{r}}^t - \vect{\bar{r}}^t } = 0 $. Together with \eqref{Intermediate}, it follows that \eqref{GradCoverged} holds.
		\end{proof}
		
		From Lemma 2, we can immediately have that $ \bar{f}^{t_j}\roundBrack{\vect{\Gamma}, \vect{q}} $ converges to $ \hat{f}\roundBrack{\vect{\Gamma}, \vect{q}} $ almost surely. This completes the proof.
	
	\section{Proof of Theorem 1}
	Denote $ \pmb{\phi}=\squareBrack{\vect{\Gamma}^\T, \vect{q}^\T}^\T $ as the composite variable and $ \bar{\pmb{\phi}}=\squareBrack{\bar{\vect{\Gamma}}^\T, \bar{\vect{q}}^\T}^\T $ as the optimal solution of the surrogate problem \eqref{SurrogateProblem}, respectively. Then we simply define $ f\roundBrack{ \vect{\Gamma},\vect{q}} $ as $ f\roundBrack{\pmb{\phi}} $ and $ \bar{f}^t\roundBrack{\vect{\Gamma}, \vect{q}} $ as $ \bar{f}^t\roundBrack{\pmb{\phi}} $, respectively. The proof of Theorem 1 can be split into three steps.
	
	1. We first prove that $ \lim\inf_{t\rightarrow\infty}\norm{\bar{\pmb{\phi}^t}-\pmb{\phi}^t}=0 $ w.p.1.
	
	Since $ \bar{f}^t\roundBrack{\pmb{\phi}} $ is strongly concave over $ \pmb{\phi} $, we have
	\begin{equation}\label{StrongConcave}
	\nabla_{\pmb{\phi}}^\T \bar{f}^t\roundBrack{\pmb{\phi}^t}\vect{d}^t\geq\eta\norm{\vect{d}^t}^2+\bar{f}^t\roundBrack{\bar{\pmb{\phi}^t}}-\bar{f}^t\roundBrack{\pmb{\phi}^t}\geq\eta\norm{\vect{d}^t}^2
	\end{equation}
	for some $ \eta>0 $, where $ \vect{d}^t=\bar{\pmb{\phi}^t}-\pmb{\phi}^t $. Moreover, the gradient $ \nabla_{\pmb{\phi}}f\roundBrack{\pmb{\phi}} $ is Lipschitz continuous, and thus there exists a constant $ L_0>0 $ such that
	\begin{equation*}
	\begin{split}
	f\roundBrack{\pmb{\phi}^{t+1}}&\geq f\roundBrack{\pmb{\phi}^t}+\gamma_t\nabla_{\pmb{\phi}}^\T f\roundBrack{\pmb{\phi}^t}\vect{d}^t-L_0\roundBrack{\gamma_t}^2\norm{\vect{d}^t}^2\\
	&=f\roundBrack{\pmb{\phi}^t}-L_0\roundBrack{\gamma_t}^2\norm{\vect{d}^t}^2 +\gamma_t\roundBrack{\nabla_{\pmb{\phi}}^\T f\roundBrack{\pmb{\phi}^t}-\nabla_{\pmb{\phi}}^\T \bar{f}^t\roundBrack{\pmb{\phi}^t}+\nabla_{\pmb{\phi}}^\T \bar{f}^t\roundBrack{\pmb{\phi}^t}}\vect{d}^t\\
	&\geq f\roundBrack{\pmb{\phi}^t}+\gamma_t\eta\norm{\vect{d}^t}^2-o\roundBrack{\gamma_t},
	\end{split}
	\end{equation*}
	where the last inequality follows from equation \eqref{StrongConcave} and $ \lim_{t\rightarrow\infty}\norm{\nabla_{\pmb{\phi}}f\roundBrack{\pmb{\phi}^t}-\nabla_{\pmb{\phi}}\bar{f}^t\roundBrack{\pmb{\phi}^t}}=0 $, which is a result of Lemma 1. Next, we show by contradiction that $ \lim\inf_{t\rightarrow\infty}\norm{\bar{\pmb{\phi}^t}-\pmb{\phi}^t}=0 $ w.p.1. Suppose $ \lim\inf_{t\rightarrow\infty}\norm{\bar{\pmb{\phi}^t}-\pmb{\phi}^t}\geq\chi >0 $ with a positive probability. Then we can find a realization such that $ \norm{\vect{d}^t}\geq\chi, \forall t $. We focus on such a realization. By choosing a sufficient large $ t_0 $, then there exists a constant $ \bar{\eta}>0 $, such that
	\begin{equation}\label{change}
	f\roundBrack{\pmb{\phi}^{t+1}}-f\roundBrack{\pmb{\phi}^t}\geq\gamma_t\bar{\eta}\norm{\vect{d}^t}^2, \forall t\geq t_0.
	\end{equation}
	Then it follows from \eqref{change} that
	\begin{equation*}
	f\roundBrack{\pmb{\phi}^t}-f\roundBrack{\pmb{\phi}^{t_0}}\geq\bar{\eta}\chi^2\sum_{j=t_0}^{t-1}\gamma_j,
	\end{equation*}
	which, in view of $ \sum_{j=t_0}^{\infty}\gamma_j=\infty $, contradicts the boundness of $ \curlyBrack{f\roundBrack{\pmb{\phi}^t}} $. Therefore, it must be $ \lim\inf_{t\rightarrow\infty}\norm{\bar{\pmb{\phi}^t}-\pmb{\phi}^t}=0 $ w.p.1.
	
	2. Then we prove that $ \lim\sup_{t\rightarrow\infty}\norm{\bar{\pmb{\phi}^t}-\pmb{\phi}^t}=0 $ w.p.1.
	
	We first prove a useful lemma.
	\begin{lemma}
		\textit{There exists a constant $ L>0 $ such that 
			\begin{equation*}
			\norm{\bar{\pmb{\phi}}^{t_1}-\bar{\pmb{\phi}}^{t_2}}\leq L\norm{\pmb{\phi}^{t_1}-\pmb{\phi}^{t_2}}+e\roundBrack{t_1,t_2},
			\end{equation*}
			where $ \lim_{t_1,t_2\rightarrow\infty}e\roundBrack{t_1,t_2}=0 $.}
	\end{lemma}
	From Lemma 2 and the Lipschitz continuity of $ f\roundBrack{\pmb{\phi}} $, we have
	\begin{equation}\label{change1}
	\abs{ \bar{f}^{t_1}\roundBrack{\pmb{\phi}}-\bar{f}^{t_2}\roundBrack{\pmb{\phi}} }\leq C\norm{\pmb{\phi}^{t_1}-\pmb{\phi}^{t_2}}+e^{'}\roundBrack{t_1,t_2},
	\end{equation}
	where $ \lim_{t_1,t_2\rightarrow\infty}e^{'}\roundBrack{t_1,t_2}=0 $ and $ C>0 $ is a constant. Obviously, the problem \eqref{SurrogateProblem} is strictly convex, when the objective function is changed by some amount $ e\roundBrack{\pmb{\phi}} $, the optimal solution $ \bar{\pmb{\phi}}^t $ will be changed by the same order, i.e., the change is within the range $ \pm O\roundBrack{\abs{e\roundBrack{\pmb{\phi}}}} $. Thus it follows from \eqref{change1} and the strong convexity of $ \bar{f}^t\roundBrack{\pmb{\phi}} $ that 
	\begin{equation}\label{change2}
	\norm{\bar{\pmb{\phi}}^{t_1}-\bar{\pmb{\phi}}^{t_2}}\leq C_1 C\norm{\pmb{\phi}^{t_1}-\pmb{\phi}^{t_2}}+C_1 e^{'}\roundBrack{t_1,t_2}
	\end{equation}
	for some constant $ C_1>0 $. Finally, Lemma 3 follows from \eqref{change2} immediately.
	
	Using Lemma 3 and the fact that $ \lim\inf_{t\rightarrow\infty}\norm{\bar{\pmb{\phi}^t}-\pmb{\phi}^t}=0 $ w.p.1., and following the same analysis as that in \cite{2016ParallelDecomposition}, Proof of Theorem 1, it can be shown that  $ \lim\sup_{t\rightarrow\infty}\norm{\bar{\pmb{\phi}^t}-\pmb{\phi}^t}=0 $ w.p.1. Therefore, we have 
	\begin{equation}\label{LimittingPoint}
	\lim_{t\rightarrow\infty}\norm{\bar{\pmb{\phi}^t}-\pmb{\phi}^t}=0, \wpone.
	\end{equation}
	
	3. We are finally ready for the proof of Theorem 1.
	
	According to Lemma 1 and equation \eqref{LimittingPoint}, the limiting point $ \curlyBrack{\vect{\Gamma}^{*}, \vect{q}^{*}} $ is the optimal solution of the following convex problem almost surely:
	\begin{equation}\label{Converged}
	\max_{\vect{\Gamma}\in\mathcal{G},\vect{q}\in\mathcal{Q}}\quad  \hat{f}\roundBrack{\vect{\Gamma},\vect{q}} 
   \end{equation}
	Thus the limiting point $ \curlyBrack{\vect{\Gamma}^{*}, \vect{q}^{*}} $ satisfies the KKT conditions of \eqref{Converged}, which are also KKT conditions of the original problem \eqref{OriginalProblem}. This completes the proof.
	\ifCLASSOPTIONcaptionsoff
	\newpage
	\fi

	\bibliographystyle{IEEEtran}
	\bibliography{IEEEabrv,Ref}

\end{document}